\documentclass[onecolumn, journal]{IEEEtran}

\usepackage{subfigure}
\usepackage{amsmath}
\usepackage{amsfonts}
\usepackage{mathrsfs}
\usepackage[dvips]{graphicx}

\usepackage{epsfig}
\usepackage{color}
\usepackage{graphicx}
\usepackage{array}
\usepackage{amssymb}
\usepackage{subfigure}

\newtheorem{thm}{Theorem}
\newtheorem{prop}{Proposition}
\newtheorem{lemma}{Lemma}

\begin{document}

\newcommand{\no}{\noindent}
\newcommand{\be}{\begin{eqnarray}}
\newcommand{\ee}{\end{eqnarray}}
\newcommand{\beeq}{\begin{equation}}
\newcommand{\eeeq}{\end{equation}}
\newcommand{\beeqs}{\begin{eqnarray*}}
\newcommand{\eeqs}{\end{eqnarray*}}
\newcommand{\bms}{\boldsymbol}
\newcommand{\expec}{\mathbf{E}}

\headheight 0in

\title{Jamming in Fixed-Rate Wireless Systems with Power Constraints - Part II: Parallel Slow Fading Channels}
\author{George T. Amariucai, Shuangqing Wei and Rajgopal Kannan}
\maketitle
\footnotetext[1]{G. Amariucai and S. Wei are with the Department of ECE, Louisiana State
University. E-mail: gamari1@lsu.edu, swei@ece.lsu.edu.}
\footnotetext[2]{R. Kannan is with the Department of CS, Louisiana State University,
E-mail: rkannan@bit.csc.lsu.edu. }

\begin{abstract}
This is the second part of a two-part paper that studies the problem of
jamming in a fixed-rate transmission system with fading. In the first
part, we studied the scenario with a fast fading channel,
and found Nash equilibria of mixed strategies for short term power
constraints, and for average power constraints with and without channel
state information (CSI) feedback. We also solved the equally important
maximin and minimax problems with pure strategies. Whenever we dealt with
average power constraints, we decomposed the problem into two levels
of power control, which we solved individually. In this second part of
the paper, we study the scenario with a parallel, slow fading channel,
which usually models multi-carrier transmissions, such as OFDM. Although
the framework is similar as the one in Part I \cite{myself3}, dealing
with the slow fading requires more intricate techniques. Unlike in
the fast fading scenario, where the frames supporting the transmission
of the codewords were equivalent and completely characterized by the
channel statistics, in our present scenario the frames are unique, and
characterized by a specific set of channel realizations. This leads to
more involved inter-frame power allocation strategies, and in some cases
even to the need for a third level of power control.  We also show that
for parallel slow fading channels, the CSI feedback helps in the battle
against jamming, as evidenced by the significant degradation to system
performance when CSI is not sent back.  We expect this degradation to
decrease as the number of parallel channels $M$ increases, until it
becomes marginal for $M\to \infty$ (which can be considered as the case
in  Part I).
\end{abstract}


\section{Introduction}

The concept of jamming plays an extremely important role in ensuring the quality and security of
wireless communications, especially at this moment when wireless networks are quickly becoming ubiquitous.
Although the recent literature covers a wide variety of jamming problems \cite{uluk1, basar3,
basar1, medard, hughes, mallikscholtz, altman}, the investigation of optimal jamming and anti-jamming
strategies for the parallel slow-fading channel is missing.

The parallel slow-fading channel is a widely used model for OFDM transmission \cite{tsevisw}.
Since the usual definition of capacity does not provide a positive performance indicator for this model, a
more adequate performance measure is the probability of outage \cite{tsevisw}, defined as the probability
that the instantaneous mutual information characterizing the parallel channel, under a given channel realization,
is below a fixed transmission rate $R$. Under the optimal diversity-multiplexing tradeoff, the parallel slow-fading
channel with $M$ subchannels is known \cite{tsevisw} to yield an $M$-fold diversity gain over the scalar single
antenna channel. However the diversity-multiplexing tradeoff only gives an approximative analytical evaluation
of the probability of outage for a given rate $R$ and a signal-to-noise ratio (SNR), and this approximation is
usually accurate only in the high SNR region. Thus, for evaluating a system which functions at a moderate SNR,
the exact probability-of-outage vs. transmission-rate curve is often computed numerically. Moreover, the high
SNR assumption is clearly not adequate for studying a practical uncorrelated jamming situation, where the jammer's
power should be considered at least comparable to the legitimate transmitter's.

Therefore, we aim at deriving the exact probability of outage achievable in the presence of a jammer, over
our parallel slow fading channel, for a fixed transmission rate $R$. Our channel model is depicted in  Figure
\ref{fig_channel_model}. The span of a codeword is denoted by ``frame''. To model our parallel slow fading channel, each
frame is divided into $M$ ``blocks'' (corresponding to the $M$ subchannels), each of which consists of $N$ channel uses,
like in Figure~\ref{fig_frames}.

\begin{figure}[b]
\centering
\includegraphics[scale=0.5]{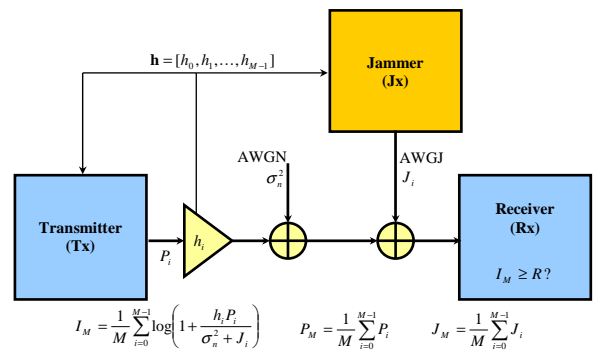}
\caption{Channel model}\label{fig_channel_model}
\end{figure}

The channel fading is slow, such that the corresponding channel coefficients remain constant over each block and vary
independently across different blocks. The channel coefficients are complex numbers, and their squared absolute values are denoted as
$h_m$. The vector $\mathbf{h}=[h_0, h_1, \ldots, h_{M-1}]$ of channel coefficients over a whole frame is assumed to be perfectly
known to the receiver, and can be made available by feedback (if the receiver wishes) to the transmitter (Tx), and jammer
(Jx) before the transmission begins.
It was shown in \cite{caire} that the feedback of channel state information (CSI) (i.e. the $M$ coefficients of a frame)
brings moderate benefits for the parallel slow-fading channel without jamming. Thus, by employing optimal power
control strategies, the transmitter can lower the probability of outage for fixed transmission rate and SNR.
In this paper, we study both the scenarios when the CSI is fed back by the legitimate receiver -- and hence
all $M$ channel coefficients characterizing a frame are available to both transmitter and jammer in a non-causal
fashion (it is only natural to assume that if the transmitter has full CSI, the jammer can get the same information
by eavesdropping) -- and the scenario when no feedback takes place and thus the CSI is only available to the receiver.

\begin{figure}
\centering
\includegraphics[scale=0.5]{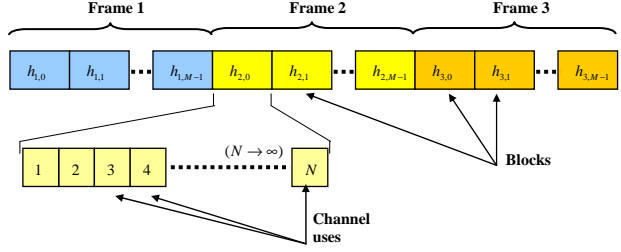}
\caption{Frames, blocks and channel uses}\label{fig_frames}
\end{figure}

In addition to fading, the transmission is affected by additive white complex Gaussian noise (AWGN), and by a jammer.
The jammer has no knowledge about the transmitter's output, or even the codebook that the transmitter is using,
and hence it deploys its most harmful strategy: it transmits white complex Gaussian noise \cite{diggavi} (AWGJ in
Figure \ref{fig_channel_model}).

The transmitter (Tx) uses a complex Gaussian codebook. Over a given frame, it allocates power $P_m$ to block $m$, $0\leq m\leq M-1$,
while the jammer (Jx) invests power $J_m$ in jamming the same block with noise.
As assumed in \cite{caire}, the number of channel uses per block is large $N\to \infty$ in order to average out the impact of
the Gaussian noise. Under these assumptions, the instantaneous mutual information characterizing a subchannel $m$ is given by
$I(h_m,P_m,J_m)=\log (1+\frac{h_m P_m}{\sigma_N^2+J_m})$, where $\sigma_N^2$ is the variance of the ambient AWGN.
The following denotations will be repeatedly used in the sequel:

\begin{itemize}
\item Power allocated by the transmitter over a
frame:\\$P_{M}=\frac{1}{M}\sum_{m=0}^{M-1}P_m$;
\item Power allocated by the jammer over  a
frame:\\$J_{M}=\frac{1}{M}\sum_{m=0}^{M-1}J_m$;
\item Instantaneous mutual information between the transmitter and the receiver over a
frame:\\$I_{M}=\frac{1}{M}\sum_{m=0}^{M-1}I(h_m,P_m,J_m)$.
\end{itemize}

Note that $P_M$ is a function of the channel realization $\mathbf{h}$, so we often write $P_M(\mathbf{h})$ when this relation
needs to be explicitly emphasized. $P_M(\mathbf{h})$  can also be interpreted as the function giving the power distribution across
different frames. We also use $P_M(h)$ and $J_M(h)$ to denote inter-frame power allocation for the case $M=1$, since in this case a frame
only contains one block. Like in \cite{myself3}, throughout this paper we shall also use the notation $c=\exp(MR)$ for simplicity.

As depicted  in Figure \ref{fig_channel_model}, our channel model is similar to that of \cite{uluk1}.
The difference, however, is that we investigate the jamming problem in slow-fading channels and  hence the
probability of outage, defined as the probability that the instantaneous mutual information
$I_{M}$ of the channel is lower than the fixed transmission rate $R$ \cite{caire} is considered as an
objective function $P_{out}=Pr(I_M<R)$ (while \cite{uluk1} assumes fast fading and uses the ergodic capacity as objective).
Our problem is still formulated as a two-player, zero-sum game. The transmitter wants to achieve reliable communication
and hence minimize the outage probability, while the jammer wants to induce outage and maximize the outage probability.
Strategies consist of varying transmission powers based on the CSI (i.e. the perfect knowledge of $\mathbf{h}$) if available,
or solely on the channel's statistics if CSI is not available.
The properties of our different objective function make our new jamming and anti-jamming problem much more challenging to solve.

It is easy to find similarities to the fixed rate system with fast fading which was studied in the first part
of this paper \cite{myself3}. In fact, the fast fading scenario of \cite{myself3} can be obtained as a particular case
of the current setup, by allowing a large number of blocks per frame $M\to \infty$ (corresponding to an infinite
number of subchannels). In doing so, the different frames are no longer characterized by their respective channel
realizations, but instead they become long enough to display the statistical properties of the channel coefficient
and thus become equivalent. This is why our present parallel slow fading scenario is more involved than the fast fading
model of Part I of this paper \cite{myself3}, especially when it comes to resolving the optimal power allocation between different frames.
Sometimes this additional complexity leads to an additional level of power control, as we shall see in Section \ref{section3}.

Our contributions are summarized below:

\begin{itemize}

\item We first investigate the case where the receiver feeds back the channel state information (CSI)
which becomes available to both transmitter and jammer. For the short-term power constraints case we show the existence
of and find a Nash equilibrium of pure strategies. Note that for a two-person, zero-sum game,
all Nash equilibria have the same value \cite{meyerson}. Since an equilibrium of pure strategies
is also an equilibrium of mixed strategies, our Nash equilibrium of pure strategies provides the
complete solution of the game.

\item For the case with long-term power constraints we find the maximin and minimax solutions of pure strategies,
and show they do not coincide (hence the non-existence of a Nash equilibrium of pure strategies).
Traditional methods of optimization, such as the KKT conditions, cannot be
applied to solve for these solutions completely. Therefore we provide a new, more intuitive approach based on
the special duality property discussed in Appendix II-D of the first part of this paper \cite{myself3}. 
As argued in \cite{myself3}, Nash equilibria of mixed strategies may not always be the best solutions to jamming
problems. A smart jammer could eavesdrop the channel and detect both the legitimate transmitter's presence and
its power level. Therefore, we believe that the maximin and minimax problem formulations with pure strategies are
of great importance in understanding and resolving the practical jamming situations
(in the worst case, they provide upper and lower bounds on the system's performance).

\item The optimal pure strategies of allocating power between frames, for the maximin and minimax
formulations, are found as the solutions of two simple numerical algorithms. These algorithms
function according to two different techniques which we explain in the sequel and we dub as
``the vase water filling problems''.

\item Mixed strategies are discussed next. We show that for completely characterizing this scenario
we need three different levels of power control. We then particularize and obtain numerical results
for the special simple case with only one block per frame ($M=1$).

\item Finally, we compare our results to the case when the channel state information is only available to the
receiver. We derive a Nash equilibrium for $M=1$, and show that unlike in the fast fading scenario (where CSI
feedback brings negligible improvements), under our current parallel slow fading channel model, perfect knowledge
about the CSI at all parties can substantially improve performance.
\end{itemize}

The paper is organized as follows. Section \ref{section1}
deals with the short term power constrained problem when full CSI is available to all parties.
Section \ref{section2} studies the scenario with long term power constraints and pure strategies
under the same assumption of available CSI. Mixed strategies are discussed in Section \ref{section3}.
For comparison purposes, Section \ref{section4} presents results for the case with no CSI feedback.
Finally, conclusions are drawn in Section \ref{section5}.


\section{CSI Available to All Parties. Jamming Game with Short-Term Power Constraints}\label{section1}

The game with short-term power constraints is the less complex of the two games we discuss in the sequel.
In this game, the transmitter's goal is to:
\begin{gather}\label{game11}
\left\{ \begin{array}{cc} \textrm{Minimize} & \mbox{Pr}(I_M(\mathbf{h},P(h),J(h))<R)\\
\textrm{Subject to} & P_M(\mathbf{h}) \leq \mathcal{P},\textrm{with prob. 1}
\end{array} \right.
\end{gather}
while the jammer's goal is to:
\begin{gather}\label{game12}
\left\{ \begin{array}{cc} \textrm{Maximize} & \mbox{Pr}(I_M(\mathbf{h},P(h),J(h))<R)\\
\textrm{Subject to} & J_M(\mathbf{h})\leq \mathcal{J},\textrm{with prob. 1}. 
\end{array} \right.
\end{gather}

We shall prove that this game is closely related to a different two player, zero-sum game,
which has the mutual information between Tx and Rx as a cost/reward function:

\begin{gather} \label{game21}
\textrm{Tx}\left\{ \begin{array}{cc} \textrm{Maximize} & I_M(\mathbf{h},P(h),J(h))\\
\textrm{Subject to} & P_M(\mathbf{h})\leq \mathcal{P},
\end{array} \right.
\end{gather}

\begin{gather}\label{game22}
\textrm{Jx}\left\{ \begin{array}{cc} \textrm{Minimize} & I_M(\mathbf{h},P(h),J(h))\\
\textrm{Subject to} & J_M(\mathbf{h}) \leq \mathcal{J}. 
\end{array} \right.
\end{gather}

This latter game is characterized by the following proposition:
\vspace*{4pt}
\begin{prop}\label{prop_short_term}
The game of (\ref{game21}) and (\ref{game22}) has a Nash equilibrium point given by the following strategies:

\begin{gather}\label{sol11}
P^*(h_m)=\left\{ \begin{array}{ccc} (\frac{1}{\eta}-\frac{\sigma_N^2}{h_m})^+ & \textrm{if} & h_m<\frac{\sigma_N^2 \eta}{1-\sigma_N^2 \nu}\\
\frac{h_m}{\eta(h_m+\frac{\eta}{\nu})} & \textrm{if} & h_m\geq\frac{\sigma_N^2 \eta}{1-\sigma_N^2 \nu}
\end{array} \right.
\end{gather}

\begin{gather}\label{sol12}
J^*(h_m)=\left\{ \begin{array}{ccc} 0 & \textrm{if} & h_m<\frac{\sigma_N^2 \eta}{1-\sigma_N^2 \nu}\\
\frac{h_m}{\nu(h_m+\frac{\eta}{\nu})}-\sigma_N^2 & \textrm{if} & h_m\geq\frac{\sigma_N^2 \eta}{1-\sigma_N^2 \nu}
\end{array} \right.
\end{gather}
 where $\eta$ and $\nu$ are constants that can be determined from the power constraints.
\end{prop}
\vspace*{4pt}

\begin{proof}
The proof is a straightforward adaptation of Section IV.B in \cite{uluk1},
and is outlined in Appendix \ref{app1}.
\end{proof}
\vspace*{4pt}

The connection between the two games above is made clear in the following theorem, the proof of which
follows in the footsteps of \cite{caire} and is given in Appendix \ref{app1}.
\vspace*{4pt}
\begin{thm}\label{thm_short_term}
Let $P^*(h)$ and $J^*(h)$ denote the Nash equilibrium solutions of the game
described by (\ref{game21}) and (\ref{game22}).
Then the original game of (\ref{game11}), (\ref{game12}) has a Nash equilibrium point, which is given by
the following pair of strategies:

\begin{gather}
\widehat{P}(h_m)=\left\{ \begin{array}{ccc} P^*(h_m) & \textrm{if} & \mathbf{h}\in \mathcal{U}(R,\mathcal{P}, \mathcal{J})\\
P_a(h_m) & \textrm{if} &\mathbf{h}\notin \mathcal{U}(R,\mathcal{P}, \mathcal{J})
\end{array} \right.
\end{gather}

\begin{gather}
\widehat{J}(h_m)=\left\{ \begin{array}{ccc} J_a(h_m) & \textrm{if} &\mathbf{h}\in \mathcal{U}(R,\mathcal{P}, \mathcal{J})\\
J^*(h_m) & \textrm{if} &\mathbf{h}\notin \mathcal{U}(R,\mathcal{P}, \mathcal{J})
\end{array} \right.
\end{gather}

where $\mathcal{U}(R,\mathcal{P}, \mathcal{J})=\{ \mathbf{h}\in \mathbb{R}_{+}^M : I_M(\mathbf{h},P^*(h),J^*(h))\geq R \}$,
and where $P_a(h)$ and $J_a(h)$ are some arbitrary power allocations satisfying the power constraints respectively.
\end{thm}
\vspace*{4pt}


\section{CSI Available to All Parties. Jamming Game with Long-Term Power Constraints: Pure Strategies}\label{section2}

The long-term power constrained jamming game can be formulated  as:
\begin{gather} \label{game31}
\textrm{Tx}\left\{ \begin{array}{cc} \textrm{Minimize} &
\mbox{Pr}(I_M(\mathbf{h},\{P_m\},\{J_m\})<R)\\
\textrm{Subject to} & E[P_M(\mathbf{h})] \leq \mathcal{P}
\end{array} \right.
\end{gather}

\begin{gather} \label{game32}
\textrm{Jx}\left\{ \begin{array}{cc} \textrm{Maximize} &
\mbox{Pr}(I_M(\mathbf{h},\{P_m\},\{J_m\})<R)\\
\textrm{Subject to} & E[J_M(\mathbf{h})]\leq \mathcal{J}
\end{array} \right.
\end{gather}
where the expectation is taken with respect to the vector of channel coefficients
$\mathbf{h}=(h_0, h_1, \ldots, h_{M-1})\in \mathbb{R}_+^M$, and
$\mathcal{P}$ and $\mathcal{J}$ are the upper-bounds on average transmission
power of the source and jammer, respectively.

Contrary to the previous short-term power constraints scenario, if long-term
power constraints are used it is possible to have $P_M(\mathbf{h})> \mathcal{P}$
for a particular channel realization $\mathbf{h}$, as long as the average of $P_M(\mathbf{h})$
over all possible channel realizations is less than $\mathcal{P}$.

Let $\mathfrak{m}$ denote the probability measure introduced by the probability density
function (p.d.f.) of $\mathbf{h}$, i.e., for a set
$\mathscr{A}\subseteq \mathbb{R}_+^M$, we have
$\mathfrak{m}(\mathscr{A})=\int_{\mathscr{A}} f(\mathbf{h})d\mathbf{h}$.
Integrating with respect to this measure is equivalent to
computing an average with respect to the p.d.f. given by $f(\mathbf{h})$, i.e.,
$d\mathfrak{m}(\mathbf{h})=f(\mathbf{h})d\mathbf{h}$.

Both transmitter and jammer have to plan in terms of power allocation, 
considering both the instantaneous realization and the probability distribution of the channel coefficient vector,
as well as their opponent's strategy.

If the number of blocks $M$ in each frame is larger than $1$, the
game between transmitter and jammer has two levels. The first
(coarser) level is about power allocation between frames, and has
the probability of outage as a cost/reward function. This is the
only level that shows up in the case of $M=1$. The second (finer)
level is that of power allocation between the blocks within a
frame.

An important comment similar to that in \cite{myself3} needs to be made.
We should point out that decomposing the problem into several (two or three) levels of power control,
each of which is solved separately,  does not restrict the generality of our solution.
In proving our main results we take a contradictory approach. That is, instead of
directly deriving each optimal strategy, we assume an optimal solution has already been reached and
show it has to satisfy a set of properties. We do this by first assuming that the properties are not satisfied,
and then showing that under this assumption at least one of the players can improve its strategy (and hence
the original solution cannot be optimal). The properties are selected such that they are not only
necessary, but also sufficient for the completely characterizing the optimal solution (i.e. there exists a unique
pair of strategies that satisfy these properties).

\subsection{Power Allocation between the blocks in a Frame}\label{ss1}

In this subsection we only deal with the second (intra-frame) level of power allocation
for the maximin and minimax problems.
The first (inter-frame) level will be investigated in detail in the following two subsections.

The probability of outage is determined by the $\mathfrak{m}$-measure of the set over which the transmitter is
not present or the jammer is successful in inducing outage.
This set is established in the first level of power control.
Note that the first level power allocation strategies cannot be derived
before the second level strategies are available.

In the maximin case (when the jammer plays first), assume that the
jammer has already allocated some power $J_M$ to a given frame.
Naturally, the transmitter knows $J_M$ (the maximin problem assumes
that the transmitter is fully aware of the jammer's strategy).
Depending on the channel realization, the value of $J_M$, and its
own power constraints, the transmitter decides whether it wants to
achieve reliable communication over that frame. If it decides to
transmit, it needs to spend as little power as possible (the
transmitter will be able to use the saved power for achieving
reliable communication over another set of positive $\mathfrak{m}$-measure, and
thus to decrease the probability of outage). Therefore, the
transmitter's objective is to minimize the power $P_M$ spent for
achieving reliable communication. The transmitter will adopt this
strategy whether the jammer is present over the frame, or not. The
jammer's objective is then to allocate $J_M$ between the blocks
such that the required $P_M$ is maximized.

In the minimax scenario (when transmitter plays first)
the jammer's objective is to minimize the power $J_M$ used for jamming the transmission over a given
frame. The jammer will only transmit if the transmitter is present with some $P_M$.
The transmitter's objective is to distribute $P_M$ between blocks
such that the power required for jamming is maximized.

The two problems can be formulated as:

\vspace*{4pt}
{\bf \emph{Problem 1}} (for the maximin solution - jammer plays first)
\begin{gather}
\max_{\{J_m\geq0\}} \Big[ \min_{\{P_m\geq0\}} P_M=\frac{1}{M}\sum_{m=0}^{M-1} P_m, \nonumber\\
~\textrm{s.t.}~  I_M(\{P_m\},\{J_m\})\geq R \Big]
\textrm{s.t.} \frac{1}{M}\sum_{m=0}^{M-1}J_m\leq J_M  ;
\label{probl1}
\end{gather}
\vspace*{4pt}

{\bf \emph{Problem 2}} (for the minimax solution - transmitter plays first)
\begin{gather}
\max_{\{P_m\geq0\}} \Big[ \min_{\{J_m\geq0\}} J_M=\frac{1}{M}\sum_{m=0}^{M-1} J_m, \nonumber\\
~\textrm{s.t.}~  I_M(\{P_m\},\{J_m\})\leq R \Big]
\textrm{s.t.} \frac{1}{M}\sum_{m=0}^{M-1}P_m\leq P_M.
\label{probl2}
\end{gather}
\vspace*{4pt}

These problems can be solved by methods very similar to those presented in the first part of this paper \cite{myself3}. For the
brevity of this presentation, we shall only point out the main results, and defer all proofs to the
Appendix \ref{app2}. The following propositions fully characterize the solutions.

\vspace*{4pt}
\begin{prop}\label{circ_pr_prop1}
The optimal solution of either of the two problems above satisfies both constraints with equality.
\end{prop}
\vspace*{4pt}

\vspace*{4pt}
\begin{prop}\label{circ_pr_thm}
(I) Take the game given by (\ref{game21}) and (\ref{game22}) and set the constraints to
$P_M(\mathbf{h})\leq P_{M,1}$ and $J_M(\mathbf{h})\leq J_{M,1}$. Denote the resulting value of the objective by
$I_M(\mathbf{h},P(h),J(h))=R_1$. Then solving \emph{Problem 1} above with the constraints
$\frac{1}{M}\sum_{m=0}^{M-1}J_m\leq J_{M,1}$ and $I_M(\{P_m\},\{J_m\})\geq R_1$ yields the objective
$P_M=P_{M,1}$.
Moreover, any pair of power allocations across blocks that makes an optimal solution
of the game in (\ref{game21}) and (\ref{game22}) is also an optimal solution of \emph{Problem 1}, and conversely.

(II)Take the game given by (\ref{game21}) and (\ref{game22}) and set the constraints to
$P_M(\mathbf{h})\leq P_{M,1}$ and $J_M(\mathbf{h})\leq J_{M,1}$. Denote the resulting value of the objective by
$I_M(\mathbf{h},P(h),J(h))=R_1$. Then solving \emph{Problem 2} above with the constraints
$\frac{1}{M}\sum_{m=0}^{M-1}P_m\leq P_{M,1}$ and $I_M(\{P_m\},\{J_m\})\leq R_1$ yields the objective
$J_M=J_{M,1}$.
Moreover, any pair of power allocations across blocks that makes an optimal solution
of the game in (\ref{game21}) and (\ref{game22}) is also an optimal solution of \emph{Problem 2}, and conversely.

(III) If $J_{M,1}$ is the value used for the second constraint in
\emph{Problem 1} above, and $P_{M,1}$ is the resulting value of the cost/reward function,
then solving \emph{Problem 2} with $P_M=P_{M,1}$ yields the cost/reward function
$J_M=J_{M,1}$. Moreover, any pair of power allocations across blocks that makes an optimal solution
of \emph{Problem 1}, should also make an optimal solution of \emph{Problem 2}, and conversely.
\end{prop}
\vspace*{4pt}

\vspace*{4pt}
\begin{prop}\label{uniqueness_prop}
The optimal solutions of \emph{Problem 1} and \emph{Problem 2} above are unique.
\end{prop}
\vspace*{4pt}

\vspace*{4pt}
\begin{prop}\label{propconcave}
(I) Under the optimal maximin second level power control strategies (\emph{Problem 1}),
the ``required'' transmitter power $P_M$ over a frame is a strictly increasing, continuous, concave and unbounded function of the power
$J_M$ that the jammer invests in that frame.

(II) Under the optimal minimax second level power control strategies (\emph{Problem 2}),
the ``required'' jamming power $J_M$ over a frame is a strictly increasing, continuous, convex and unbounded function of the power
$P_M$ that the transmitter invests in that frame.
\end{prop}
\vspace*{4pt}

Although under the same transmitter/jammer frame power constraints $P_M$ and $J_M$ the second level optimal power allocation
strategies for the maximin and minimax problems coincide, this result should not be associated with the notion of Nash equilibrium,
since the two problems solved above do not form a zero-sum game,
while for the game of (\ref{game31}) and (\ref{game32}), first level power control
strategies are yet to be investigated.

As in \cite{myself3}, we shall henceforth denote the function that gives the ``required'' transmitter power $P_M$ over a frame
where the jammer invests power $J_M$ by $\mathscr{P}_M(J_M, \mathbf{h})$ and its ``inverse'', i.e. the function that gives the ``required''
jamming power over a frame where the transmitter invests $P_M$ by $\mathscr{J}_M(P_M, \mathbf{h})$.
Note that unlike in \cite{myself3}, these functions are now also dependent on the channel realization $\mathbf{h}$.
A particular channel realization can be characterized in terms of the second level power allocation technique.
For instance, considering the maximin problem, we can map each channel vector $\mathbf{h}$ to a unique curve
$\mathscr{P}_M(J_M)$ in the plane. That is, for fixed $\mathbf{h}$, we increase the jamming power allocated to the frame
from $0$ to $\infty$, and compute the transmitter power $\mathscr{P}_M(J_M,\mathbf{h})$ required for achieving reliable communication.
We have already mentioned that, for any fixed $\mathbf{h}$, $\mathscr{P}_M(J_M)$ is a strictly increasing,  continuous,
concave and unbounded function. 

Next we take a closer look at the $\mathscr{P}_M(J_M, \mathbf{h})$ curves.
By inspecting the proofs of Propositions \ref{circ_pr_prop1} - \ref{propconcave}, we notice that
$j$ denotes the index of the first block on which the jammer allocates nonzero power,
while $p$ is the index of the first block on which the transmitter allocates nonzero power
(the blocks are indexed in increasing order of their squared channel coefficients $h_m$, and
both transmitter and jammer allocate more power to blocks with larger values of $h_m$).
Note also that $p\leq j$. If for a given $\mathbf{h}$ we have $p=j$ over an interval of $J_M$,
then the $\mathscr{P}_M(J_M)$ curve is linear over that interval. However, if $p<j$, the curve is
strictly concave.

We can think of the $\mathscr{P}_M(J_M)$ curve that characterizes a given channel realization $\mathbf{h}$
as being ``built'' in the following manner. We increase the jamming power allocated to the corresponding
frame, starting from $J_M=0$. We already know that without the jammer's presence the transmitter transmits
over the ``best'' blocks , i.e. the ones having the largest channel coefficients. Even as the
jammer starts interfering, its optimal strategy is such that the blocks with the largest coefficients remain
the most attractive for the transmitter. However, they do become worse than before. Hence, if without the presence of
the jammer the transmitter would normally ignore some of the blocks, as the jammer's power increases, those blocks
may slowly become more attractive. At some point, the transmitter will choose to increase the number of blocks
over which it allocates non-zero power (i.e. decrease $p$). Similarly, as the jammer's power $J_M$ increases,
the jammer moves from the best block to the best two blocks, and so on (i.e. the jammer decreases $j$).

The transmitter's and the jammer's transitions do not have to be simultaneous. Recall that the relationship between
the values of $p$ and $j$ decide whether the $\mathscr{P}_M(J_M)$ curve is linear or strictly concave over an interval
of $J_M$. Therefore, we expect the $\mathscr{P}_M(J_M)$ curves to look like a concatenation of linear and strictly
concave segments, as in Figure \ref{fig_curves}. As $J_M$ increases, the transmitter decreases the value of $p$
whenever the slope of the $\mathscr{P}_M(J_M)$ curve can be decreased by this move and similarly, the jammer
decreases the value of $j$ whenever the slope can be increased. In other words,
as $J_M$ increases, the transitions from linear portions to nonlinear
portions are caused by the transmitter, while the transitions from nonlinear to linear ones  are caused by the jammer.

\begin{figure}
\centering
\includegraphics[scale=0.8]{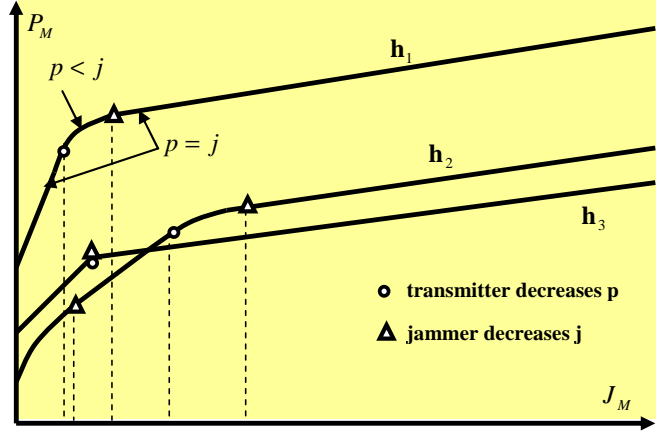}
\caption{Typical $P_M(J_M)$ curves, for different channel realizations}\label{fig_curves}
\end{figure}

In the remainder of this subsection we provide the simplest example of optimal power allocation between the
blocks of a frame. Namely, we look at the case when $M=2$ --  only two blocks per frame.

\vspace*{4pt}
{\bf \emph{Particular case: $M=2$}} \vspace*{4pt}

The case of $M=2$ is the simplest and most intuitive illustration of the second-level power control strategy.
Since we have already discussed the nice dual property between the second level minimax and maximin strategies, the following
considerations refer to the maximin scenario only. The jamming power $J_M$ has to be allocated between the two blocks
in a way that maximizes the transmitter's expense, should it decide to achieve reliable communication over the frame.
The jammer and the transmitter can each transmit over either one or both blocks.
All possible situations are considered next.

Let the two channel coefficients be $h_0\leq h_1$, and denote the transmitter's and jammer's powers allocated
to the blocks by $P_0,P_1$ and $J_0,J_1$ respectively. Also denote $x_i=J_i+\sigma_N^2$, for $i\in \{0,1\}$, and $c=\exp(2R)$.
If we take a closer look at the solutions (\ref{sol11}) and (\ref{sol12}) of the game in (\ref{game21}) and (\ref{game22}),
and if we recall that the solutions of either of our maximin and minimax second layer power allocation strategies have
a similar form (up to the constants $\eta$ and $\nu$), it is easy to observe that $x_0\leq x_1$ and $\frac{x_0}{h_0}\geq\frac{x_1}{h_1}$.
This fact is also noted in Appendix \ref{app23}, where the solution of \emph{Problem 1} is given again, with the new notation
$\lambda=1/\eta$ and $\mu=\nu/\eta$. Throughout the rest of this subsection we shall refer to the notation in Appendix \ref{app23}
and the solution in (\ref{solution1}) and (\ref{solution2}).

If the transmitter is active over both blocks, then the constraint $I_M=R$ yields
\be
\left(1+\frac{h_0}{x_0}P_0\right)\left(1+\frac{h_1}{x_1}P_1\right)=c,
\ee
and with (\ref{Pmfirstexpr}) in Appendix \ref{app23} we obtain $\lambda=\sqrt{c\frac{x_0}{h_0}\frac{x_1}{h_1}}$.

Suppose that the jammer is only present on one block of the frame, then that is the block with coefficient $h_1$.
This implies $x_0=\sigma_N^2$, and $x_1=(2J_M+\sigma_N^2)$. Under these assumptions, the transmitter will only transmit
on the first block, (that is $P_0=2P_M$ and $P_1=0$) if and only if
\begin{gather}
\lambda=\sqrt{c\frac{x_0}{h_0}\frac{x_1}{h_1}}<\frac{x_0}{h_0},
\end{gather}
which translates to $c\frac{(2J_M+\sigma_N^2)}{h_1}<\frac{\sigma_N^2}{h_0}$.

Otherwise, the transmitter is present over both blocks, performing water-pouring as in
(\ref{Pmfirstexpr}), with
\begin{gather}
\lambda=\sqrt{c\frac{(2J_M+\sigma_N^2)\sigma_N^2}{h_0 h_1}}.
\end{gather}

Note that the transmitter cannot be present only on the second block.

If the jammer decides to allocate non-zero power over both blocks,
its optimal strategy is such that $x_0/h_0\geq x_1/h_1$.
If we also have $x_0/h_0\leq c(x_1/h_1)$ (corresponding to $\lambda\geq x_0/h_0$), then the transmitter is present over both blocks.
In this case, we can particularize (\ref{Pmfirstexpr}) to $M=2$ and obtain:
\begin{gather}\label{PMis2}
P_m=\sqrt{c\frac{x_0}{h_0}\frac{x_1}{h_1}}-\frac{x_m}{h_m},~\textrm{for}~ m\in\{0, 1 \}.
\end{gather}
Define the ratio $r=\frac{x_0/h_0}{x_1/h_1}$.
Since $x_0+x_1=2(J_M+\sigma_N^2)$, we can write
\begin{gather}\label{pm2}
P_M=\frac{(J_M+\sigma_N^2)(2\sqrt{cr}-r-1)}{h_0r+h_1},~\textrm{if}~c\frac{x_1}{h_1}\geq\frac{x_0}{h_0}.
\end{gather}
Setting the derivative of $P_M$ with
respect to $r$ equal to zero, we get the unique solution
\begin{gather}
r_{opt}=\left(\frac{\sqrt{(h_1-h_0)^2+4h_0h_1c}-(h_1-h_0)}{2h_0\sqrt{c}}\right)^2,
\end{gather}
which provides the optimal allocation of the jamming power $J_M$ between the two blocks.
The value of $r_{opt}$ is between $1$ (for $h_0=h_1$) and $c$ (for $h_0=0$).
Furthermore, $P_M(r)$ is strictly increasing for $r\in [1, r_{opt})$ and strictly
decreasing for $r\in (r_{opt}, c]$, hence $r_{opt}$ is the maximizing argument in (\ref{pm2}).

This also implies that if $r_{opt}\frac{(2J_M+\sigma_N^2)}{h_1}<\frac{\sigma_N^2}{h_1}$,
the jammer's optimal strategy is to allocate all of its power to the second block.
If, on the other hand, $r_{opt}\frac{(2J_M+\sigma_N^2)}{h_1}\geq\frac{\sigma_N^2}{h_1}$, then
the jammer's best strategy is to allocate the power $J_M$ such that
the ratio $r=(x_0/h_0)/(x_1/h_1)$ equals the optimal ratio $r_{opt}$.

The remarks above conclude in the following algorithm:
\begin{itemize}
\item If $c\frac{(2J_M+\sigma_N^2)}{h_1}\leq\frac{\sigma_N^2}{h_0}$, both transmitter and jammer will
only transmit on the second block.
\item If $c\frac{(2J_M+\sigma_N^2)}{h_1}>\frac{\sigma_N^2}{h_0}$ but
$r_{opt}\frac{(2J_M+\sigma_N^2)}{h_1}\leq\frac{\sigma_N^2}{h_1}$, the jammer will allocate all its power
to the second block, while the transmitter will transmit on both blocks.
\item If $r_{opt}\frac{(2J_M+\sigma_N^2)}{h_1}>\frac{\sigma_N^2}{h_1}$, the jammer will transmit over both blocks
such that $(x_0/h_0)/(x_1/h_1)=r_{opt}$, and the transmitter will also be present on both blocks. 
\end{itemize}


\subsection{Inter-Frame Power Allocation}\label{ss4}

In this subsection we present the first level optimal power allocation
strategies.

\vspace*{4pt}
{\bf \emph{The Maximin Solution}} \vspace*{4pt}

Under our full CSI, average power constraints scenario, the jammer needs to find the best choice
of the set $\mathscr{X} \subset \mathbb{R}_+^M$ of channel realizations
over which it should be present, and the optimal way $J_M(\mathbf{h})$ to distribute
its power over $\mathscr{X}$, such that when the
transmitter employs its optimal strategy, the probability of
outage is maximized.

We already know that given the jammer's strategy, the optimal way of allocating the
transmitter's power is such that reliable communication is first obtained on the frames
that require the least amount of transmitter power. 
\begin{figure}[b]
\centering
\includegraphics[scale=0.55]{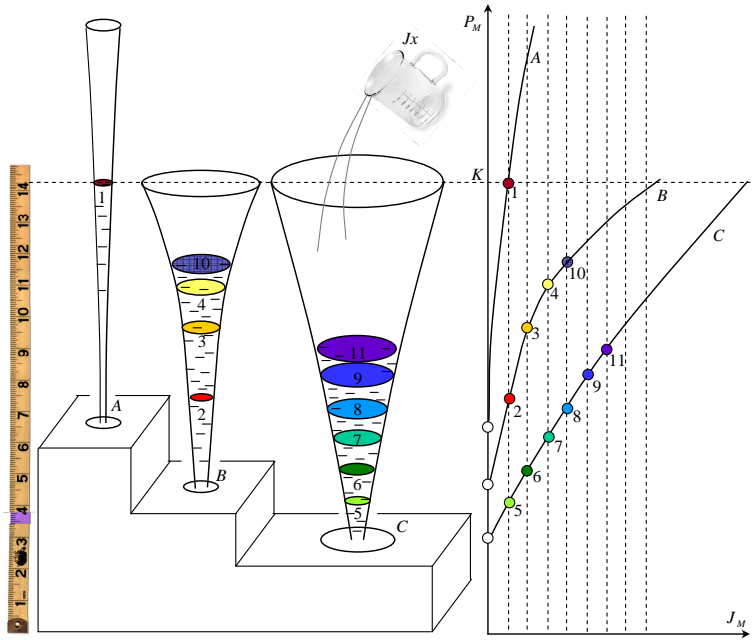}
\caption{Maximin vase filling.}\label{alg1fig}
\end{figure}
The jammer's optimal strategy is presented in Theorem \ref{thm3_long_term} below.
The theorem is complemented by the numerical algorithm and the intuition-building analogy that follows its proof.

\vspace*{4pt}
\begin{thm}\label{thm3_long_term}
It is optimal for the jammer to
make $J_M(\mathbf{h})$ satisfy  the power constraint with equality.
The optimal jammer strategy for allocating power across frames is to increase the \emph{required}
transmitter power, starting with those frames whose channel realizations exhibit
the steepest instantaneous slope of the characteristic $\mathscr{P}_M(J_M)$ curve.
The jamming power should be allocated such that the \emph{required} transmitter power
over each channel realization where the jammer is present
does not exceed a pre-defined level $K$.

The optimal value for $K$ that maximizes the outage probability can be found 
numerically, by exhaustive search in a compact interval of the positive real line.
\end{thm}
\vspace*{4pt}

\begin{proof}
Our proof takes a contradictory approach. Instead of deriving the optimal strategy defined above in a
direct manner, we show instead that any other strategy not satisfying the theorem's requirements is suboptimal.
Let $\mathscr{S}, \mathscr{X}\subset \mathbb{R}_+^M$ denote the sets of channel realizations over which
the transmitter and the jammer are present, respectively.

Suppose the jammer picks a certain strategy $J_M(\mathbf{h})$.
Since the transmitter's strategy is predictable, the jammer already knows the transmitter's
optimal strategy. Under this optimal strategy, the transmitter picks a set of frames $\mathscr{S}$ over
which it will invest non-zero power. This choice also results in a maximum level of \emph{required} transmitter  power that
will actually be matched by the transmitter. Denote this level by $K$.

Since the transmitter's strategy is the optimal response to the jammer's strategy, the \emph{required} transmitter
power should be larger than or equal to $K$ over the set of frames $\mathscr{X} \setminus \mathscr{S}$
where the jammer jams, but the transmitter does not afford to transmit. Otherwise, the transmitter
would be wasting power and its strategy would not be optimal. 

But since the jammer knows the transmitter's strategy, and knows that the transmitter will not transmit over 
$\mathscr{X} \setminus \mathscr{S}$, its optimal strategy should make the \emph{required} transmitter
power over $\mathscr{X} \setminus \mathscr{S}$ at most equal to $K$. Otherwise the jammer would be wasting power. 

We have seen how the jammer's power should be distributed over $\mathscr{X} \setminus \mathscr{S}$.
Next we show that if the jammer's power allocation over $\mathscr{S}\bigcap \mathscr{X}$ is not done
according to the theorem, the jammer's strategy is not optimal. For this, we assume that the jammer's
strategy does not satisfy the theorem's requirements, and provide a method of improvement
(i.e. we prove sub-optimality).

If the theorem is not satisfied, than there exist two sets $\mathscr{A},\mathscr{B} \subset \mathscr{S}\bigcap \mathscr{X}$
of non-zero $\mathfrak{m}$-measure such that$\frac{dP_M(\mathbf{h}_1)}{dJ_M} > \frac{dP_M(\mathbf{h}_2)}{dJ_M} ~
\forall~\mathbf{h}_1\in \mathscr{A} ~\textrm{and}~\mathbf{h}_2\in \mathscr{B}$, and such that the required $P_M$
is less than $K$ on $\mathscr{A}$ and $J_M>0$ on $\mathscr{B}$.

Consider a small enough amount of jamming power $\delta J_M$, such that, for any channel realization
$\mathbf{h}\in \mathscr{A} \bigcup \mathscr{B}$, we can modify the jamming power by $\delta J_M$ without changing the slope
of the $\mathscr{P}_M(J_M)$ curve. Subtracting $\delta J_M$ from all frames in $\mathscr{B}$, the jammer obtains
the excess power $\delta J_M m(\mathscr{B})$, which it can allocate uniformly over $\mathscr{A}$.
The jammer's total average power remains unchanged. However, the required transmitter power over
$\mathscr{A} \bigcup \mathscr{B}$  is increased (because the slopes of the $\mathscr{P}_M(J_M)$ curves corresponding to
$\mathscr{A}$ are all larger than the slopes of the $\mathscr{P}_M(J_M)$ curves corresponding to $\mathscr{B}$),
and thus the modification results in a larger probability of outage.

There exists a closed interval $[0,K_{max}] \in \mathbb{R}_+$ which includes the optimal value of $K$.
This observation is vital to the existence of a numerical algorithm that searches for the optimal $K$.
Once such an interval has been set, we can fix the desired resolution and calculate the numerical complexity
of the algorithm. We next show how the upper limit $K_{max}$ of this interval can be found. Consider the
set of channel realizations
$\mathscr{S}_0$ where the transmitter is active when the jammer does not interfere with the transmission.
Next, find the value $K_{max}$ for which, when the jammer allocates its power $\mathcal{J}$ according to the
rules of the theorem, we obtain a set $\mathscr{X}_0\subset \mathbb{R}^M_+\setminus\mathscr{S}_0$.
This means that the jammer's strategy under any $K\geq K_{max}$ has no influence upon the transmitter's strategy.
Note that such a finite $K_{max}$ can be found whenever $\mathbb{R}^M_+\setminus\mathscr{S}_0$ has non-zero
$\mathfrak{m}$-measure.
\end{proof}
\vspace*{4pt}

The algorithm in Table \ref{table1} which we used in generating our numerical results in Subsection \ref{ss5}
helps shed more light into the practicality of Theorem \ref{thm3_long_term}.
In the description of the algorithm, we assume discrete jamming power levels $J_M^k$ with $k=0, 1, \ldots$
and $J_M^0=0$, as well as a discrete
and finite channel coefficient space. As a consequence, there exists a finite number of $\mathscr{P}_M(J_M)$ curves,
each characterizing one possible channel realization, and each completely determined by a finite vector whose components are
the values of $\mathscr{P}_M(J_M^k)$ for that particular channel realization.

An intuitive description of the technique is given in Figure \ref{alg1fig}. Consider the problem where the
jammer has to pour water in a number of vases (a vase for each possible channel realization). The shape of each
vase is such that the vertical section of its wall produces a concave curve similar to the corresponding
$\mathscr{P}_M(J_M)$ curve. The jammer can afford to spend a certain volume of water. The jammer wants to
``annoy'' the transmitter, which is deeply concerned with \emph{the sum of the heights} that the water levels reach
in the vases. Hence, the jammer tries to use its available volume of water, such that the sum of the water levels'
heights is maximized. However, the jammer cannot pour all the water in the thinnest vase, because then the
transmitter might just ignore that vase. Instead, the jammer has to set a height limit $K$ which it should not exceed.
The jammer pours the water a cup at a time, starting with the vase in which a cup of water rises the water level
the quickest. In Figure \ref{alg1fig}, the order of adding cups to the vases is shown by numerals from $1$ to $11$.
The first cup is poured into the thinnest vase (vase $A$) and incidentally reaches the level $K$. Thus, no more
water should be added to vase $A$. The next three cups are added to vase $B$, and then the next five cups to
vase $C$. Then the jammer returns to vase $B$, and adds another cup, for this increases the water level more than
it would increase the level in vase $C$. Finally, the last available cup is added to vase $C$. 
The way the numerical algorithm works is illustrated in the right part of Figure \ref{alg1fig}.

\begin{table}\caption{Numerical algorithm for deriving the maximin solution.}\label{table1}
\begin{small}
\begin{tabular}{|p{8cm}|}
\hline
\noindent
Let $\mathbf{P}$ denote a matrix with each row representing
one of the vectors $\mathscr{P}_M(J_M^k)$, for different channel realizations
$\mathbf{h}$.
Let $P_{req}$ be the vector of required powers for the different frames.
The initial $P_{req}$ is set equal to the first column of $\mathbf{P}$. 
Let $K_{max}$ be the upper limit when searching for the optimal $K$.\\
Initialize $K=0$.\\
\emph{while} $K\leq K_{max}$\\
~~~~~ $p_T=0$.\\
~~~~~ Let $L$ be an index vector, the same size as $P_{req}$.\\
~~~~~ Initialize all components of $L$ to be equal to $1$.\\
~~~~~ We have the relationship $P_{req}(j)=\mathbf{P}(j,L(j))$.\\
\% Jx strategy:\\
The amount of jamming power spent at each step is accumulated into the variable $J_c$.\\
~~~~~ \emph{while} Jx power constraint is satisfied ($J_c\leq \mathcal{J}$)\\
~~~~~~~~~~ Find row $j$ of $\mathbf{P}$ with the largest difference\\
~~~~~~~~~~ between components $L(j)+1$ and $L(j)$,\\
~~~~~~~~~~ and such that $\mathbf{P}(j,L(j)+1)\leq K$.\\
~~~~~~~~~~ $P_{req}(j)=\mathbf{P}(j,L(j)+1)$.\\
~~~~~~~~~~ $L(j)=L(j)+1$.\\
~~~~~~~~~~ Weigh $J_M^j$ by probability of row $j$ and add to $J_c$.\\
~~~~~ \emph{end}\\

\% Tx strategy (Tx picks frames where required power is minimum first)\\
The amount of transmitter power spent at each step is simulated into the variable $P_c$.\\
~~~~~ \emph{while} Tx power constraint is satisfied ($P_c\leq \mathcal{P}$)\\
~~~~~~~~~~ Pick the least component of $P_{req}$.\\
~~~~~~~~~~ Add probability of corresponding frame to $p_T$.\\
~~~~~~~~~~ Add value of component, weighted by\\
~~~~~~~~~~ probability above, to $P_c$.\\
~~~~~~~~~~ Delete component from $P_{req}$.\\
~~~~~ \emph{end}\\
~~~~~ $P_{out}(K)=1-p_T$\\
~~~~~ Increment K.\\
\emph{end}\\
Select K that produces the largest $P_{out}$.\\
\hline
\end{tabular}
\end{small}
\end{table}

\vspace*{4pt}
{\bf \emph{The Minimax Solution}} \vspace*{4pt}

In Theorem \ref{circ_pr_thm} we showed that given the transmitter's and the jammer's powers $P_M$ and $J_M$
allocated to a frame, the optimal strategies for distributing these powers inside the frame are identical for
the minimax and the maximin problems.
Hence, by rotating the $\mathscr{P}_M(J_M)$ plane, we get the characteristic $\mathscr{J}_M(P_M)$ curves for
the minimax problem.

We already know that given the transmitter's strategy, the optimal way of allocating the
jammer's power is such that outage is first induced on the frames that require the least amount of jamming power. 

The transmitter's optimal strategy is presented in the following theorem, which is complemented by the numerical
algorithm and the analogy that follows its proof.
\vspace*{4pt}
\begin{thm}\label{thm4_long_term_M>1}
It is optimal for transmitter to make   $P_M(\mathbf{h})$ satisfy
the long-term power constraint with equality.
The optimal transmitter power allocation across frames
is to increase the \emph{required} jamming power up to some pre-defined level $K$,
starting with those frames on which the required transmitter power to achieve this goal
is least.
 
The optimal value for $K$ that minimizes the outage probability can be found 
numerically by exhaustive search.
\end{thm}
\vspace*{4pt}
\begin{proof}
As in the case of Theorem \ref{thm3_long_term}, we take a contradictory approach.
Instead of directly deriving the optimal strategy defined above, we show that any other strategy
not satisfying the theorem's requirements is suboptimal. Recall that $\mathscr{S}~ \textrm{and}~
\mathscr{X}\subset \mathbb{R}_+^M$ denote the sets of channel realizations over which
the transmitter and the jammer are present, respectively.

Suppose the transmitter picks a certain strategy $P_M(\mathbf{h})$.
Since the jammer's strategy is predictable, the transmitter already knows the jammer's
optimal strategy. Under this optimal strategy, the jammer should pick a set of frames $\mathscr{X}$ over
which it will invest non-zero power. This choice also results in a maximum level of \emph{required} jamming  power that
will actually be matched by the jammer. Denote this level by $K$.

Since the jammer's strategy is optimal, the \emph{required} jamming power outside the set $\mathscr{X}$
should be larger than or equal to $K$. Otherwise, the jammer would be wasting power and hence its strategy would not be optimal. 

But since the transmitter knows the jammer's strategy, it also knows that the jammer will not be present over 
$\mathscr{S} \setminus \mathscr{X}$, so the transmitter should make the \emph{required} jamming
power over $\mathscr{S} \setminus \mathscr{X}$ at most equal to $K$. Otherwise the transmitter would be wasting power. 
Hence, over $\mathscr{S} \setminus \mathscr{X}$ the transmitter should allocate power such that the required jamming power
is equal to $K$.

Next we show that if the transmitter's power allocation over $\mathscr{S}\bigcap \mathscr{X}$ is not done
according to the theorem, the transmitter's strategy is not optimal. For this, we assume that the transmitter's
strategy does not satisfy the theorem's requirements, and provide a method of improvement
(i.e. we prove sub-optimality).

If the theorem is not satisfied, than there exist two sets $\mathscr{A},\mathscr{B} \subset \mathscr{S}\bigcap \mathscr{X}$
of non-zero $\mathfrak{m}$-measure such that $P_M(\mathbf{h}_1,K)< P_M(\mathbf{h}_2,K) ~ \forall~\mathbf{h}_1\in \mathscr{A}
~\textrm{and}~\mathbf{h}_2\in \mathscr{B}$, and such that the required $J_M$ is less than $K$ on $\mathscr{A}$ and
$J_M>0$ on $\mathscr{B}$ cannot be part of the minimax solution.
Denote the original transmitter power allocation functions over $\mathscr{A}$ and $\mathscr{B}$
by $P_{M,0}^A(\mathbf{h})$ and $P_{M,0}^B(\mathbf{h})$ respectively.

For any $\mathbf{h}_1 \in \mathscr{A},~\mathbf{h}_2 \in \mathscr{B} $ and $J_{M,1}, J_{M,2}<K$, we have: 
\begin{eqnarray} \label{multineq}
\frac{K-J_{M,1}}{P_M(\mathbf{h}_1,K)-P_M(\mathbf{h}_1,J_{M,1})} \stackrel{a)}{\geq} \frac{K}{P_M(\mathbf{h}_1,K)} >{}\nonumber\\
{}\stackrel{b)}{>}\frac{K}{P_M(\mathbf{h}_2,K)} \stackrel{c)}{\geq} \frac{J_{M,2}}{P_M(\mathbf{h}_2,J_{M,2})},
\end{eqnarray}
where both $a)$ and $c)$ follow from the convexity of $\mathscr{J}_M(P_M)$ -- Proposition \ref{propconcave} -- and
$b)$ follows from the assumption in the beginning of this proof.

If the transmitter cuts off transmission over  a subset $\mathscr{B}'\subset \mathscr{B}$, it obtains
the excess power $\int_{\mathscr{B}'} P_M(\mathbf{h}) dm(\mathbf{h})$,
which it can allocate to a subset $\mathscr{A}' \subset \mathscr{A}$ such that
the required $J_M$ is equal to $K$ over $\mathscr{A}'$, i.e.
\begin{gather}
\int_{\mathscr{B}'} P_{M,0}^B(\mathbf{h}) dm(\mathbf{h})=\int_{\mathscr{A}'} \left[P_M(\mathbf{h},K)- 
P_{M,0}^A(\mathbf{h})\right]dm(\mathbf{h})
\end{gather}

Replacing $P_M(\mathbf{h}_1,J_{M,1})$ by $P_{M,0}^A(\mathbf{h})$ and
$P_M(\mathbf{h}_2,J_{M,2})$ by $P_{M,0}^B(\mathbf{h})$ in (\ref{multineq}), we see
the transmitter improves its strategy by forcing the jammer to allocate more power to the set
$\mathscr{A} \bigcup \mathscr{B}$, and hence decreases the probability of outage.
Note that since $\mathscr{B}'\subset \mathscr{S}\bigcap \mathscr{X}$, the set $\mathscr{B}'$ is in outage, regardless of whether the
transmitter is present or not. Thus, transmitter does not increase $P_{out}$ by cutting off transmission on $\mathscr{B}'$.

There exists a closed interval $[0,K_{max}] \in \mathbb{R}_+$ which includes the optimal value of $K$.
As in the maximin case, the existence of such a closed interval is required for constructing a numerical algorithm
that searches for the optimal $K$. The upper limit $K_{max}$ of this interval can be found and updated as follows.
First solve the problem for an arbitrarily chosen $K_0$, and determine the set $\mathscr{S}_0\setminus \mathscr{X}_0$ over
which the transmitter achieves reliable communication. We can set $K_{max}$ equal to the value of $K$ that yields a set
$\mathscr{S}$ of the same $\mathfrak{m}$-measure as the set $\mathscr{S}_0\setminus \mathscr{X}_0$. Note that if $K$ is increased
over this $K_{max}$, the outage probability is at least as large as that obtained for $K=K_0$ (and hence $K_0$ is a better choice).
\end{proof}
\vspace*{4pt}

\begin{figure}[b]
\centering
\includegraphics[scale=0.55]{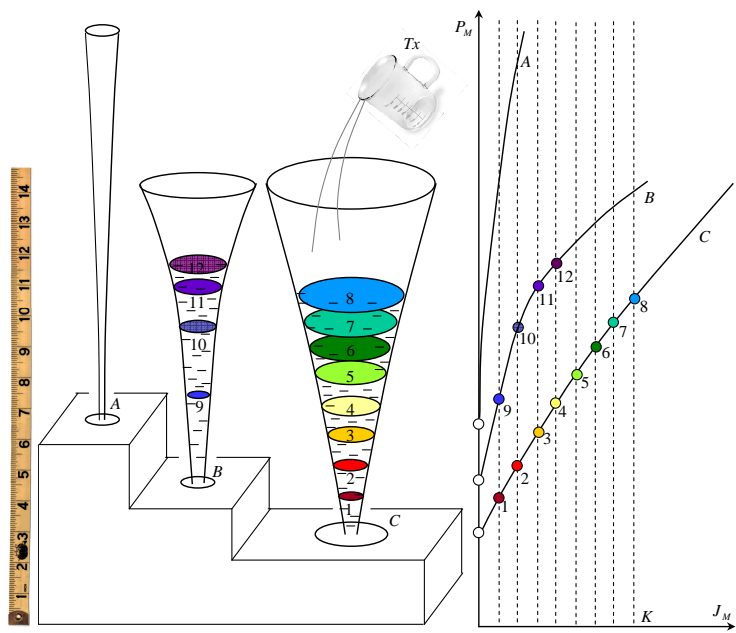}
\caption{Minimax vase filling.}\label{alg2fig}
\end{figure}

The algorithm in Table \ref{table2} which we used for our numerical results in Subsection \ref{ss5}
illustrates the application of Theorem \ref{thm4_long_term_M>1}.
In the description of the algorithm, we assume discrete jamming power levels $J_M^k$ with $k=0, 1, \ldots$
and $J_M^0=0$, as well as a discrete
and finite channel coefficient space. As a consequence, there exists a finite number of $\mathscr{P}_M(J_M)$ curves,
each characterizing one possible channel realization, and each completely determined by a finite vector whose components are
the values of $\mathscr{P}_M(J_M^k)$ for that particular channel realization.

A description of the technique is given in Figure \ref{alg2fig}, using the same vase analogy as in the maximin case.
This time, the transmitter does the pouring. Its obsession with the sum of the heights of the water levels imposes
a constraint on this sum. Under this constraint, the transmitter wants to use as much of the jammer's water as possible.
That is, the transmitter attempts to maximize the volume of water that can be accommodated by the vases, under the
constraint that the sum of the water levels' heights is less than some given value. Moreover, if the transmitter pours
water only in the thickest vase, it might not feel that it did enough damage to the jammer. Thus, the transmitter
needs to set a limit $K$.
The optimal strategy is to fill (up to volume level $K$) the thickest vase first (note that ``thickest'' refers to the fact
that when filled up to volume level $K$, the vase displays the lowest water level height, thus ``thickest'' is defined
with respect to $K$). The order in which the transmitter adds cups of water to the vases is depicted in Figure \ref{alg2fig}
by numerals from $1$ to $12$. The way the numerical algorithm works is illustrated in the right part of Figure \ref{alg2fig}.

\begin{table}\caption{Numerical algorithm for deriving the minimax solution.}\label{table2}
\begin{small}
\begin{tabular}{|p{8cm}|}
\hline
\noindent
Let $\mathbf{P}$ denote the matrix with rows representing the $\mathscr{P}_M(J_M^k)$
vectors for different channel realizations $\mathbf{h}$.
Let $K_{max}$ be value where searching for the optimal $K$ stops.\\
Initialize $K=0$.\\
\emph{while} $K\leq K_{max}$\\
\% Tx strategy:\\
The amount of transmitter power spent at each step is accumulated into the variable $P_c$.\\
~~~~~ Initialize $K=J_M^k$.\\
~~~~~ Initialize $P_c=0, ~p_T=0$.\\
~~~~~ \emph{while} Tx power constraint is satisfied ($P_c\leq \mathcal{P}$)\\
~~~~~~~~~~ Find row $j$ of $\mathbf{P}$ with least $k$-th component.\\
~~~~~~~~~~ Add probability of row $j$ to $p_T$.\\
~~~~~~~~~~ Add value of the $k$-th component, weighted\\
~~~~~~~~~~ by the probability above, to $P_c$.\\
~~~~~~~~~~ Delete row $j$ from matrix $\mathbf{P}$.\\
~~~~~ \emph{end}\\
\% Jx strategy (Jx jams frames where Tx is present,
randomly, until it reaches its power constraints):\\
~~~~~ $p_J=\frac{\mathcal{J}}{K}$.\\
~~~~~ $P_{out}(K)=p_T-p_J$.\\
~~~~~ Increment $K$.\\
\emph{end}\\
Select K that produces the least $P_{out}$.\\
\hline
\end{tabular}
\end{small}
\end{table}

\vspace*{4pt}
{\bf \emph{Particular case: $M=1$}} \vspace*{4pt}

For this simple scenario, there is no second level of power allocation.
All frames consist of only one block, and the $P_M(J_M)$ curves have the particular
affine form with parameter $h$ (the squared channel coefficient corresponding to this block):
\be
P_M=\frac{\exp(R)-1}{h}(J_M+\sigma_N^2).
\ee
Since the slopes of the $P_M(J_M)$ curves are constant with $J_M$ and the frames
with smaller values of the channel coefficients have larger characteristic slopes, we can easily
particularize Theorems \ref{thm3_long_term} and \ref{thm4_long_term_M>1}.

With the same notation  $\mathscr{X}\subset \mathbb{R}_+$ for the set of channel
realizations  over which the jammer invests non-zero power and
$\mathscr{S} \subset \mathbb{R}_+$ for the set of channel
realizations over which the transmitter uses non-zero power, we can now define the optimal
power allocation strategies.

For the maximin scenario, The jammer should deploy some $J_M(h)$ over $\mathscr{X}$ such that
the {\em required} $P_M(h)$ is constant over the whole interval $\mathscr{X}$. The purpose of
the jammer being active over $\mathscr{X} \setminus \mathscr{S}$ is to ''intimidate'' the
transmitter. The transmitter plays second, and hence takes advantage of the
jammer's weaknesses. It always chooses to be active on the subset
of $\mathscr{X}$ on which the {\em required} $P_M(h)$ is least.
This is why the optimal jammer strategy is to display no weakness,
i.e. to make $P_M(h)$ constant over $\mathscr{X}$.
These considerations are formalized in Proposition \ref{thm2_long_term} below.

\vspace*{4pt}
\begin{prop}\label{thm2_long_term}
In the maximin scenario, the jammer should adopt such a strategy as to make the transmitter's
best choice of $\mathscr{S}$ intersect $\mathscr{X}$ on the 
the left-most part of $\mathscr{S}$, and the required transmitter power
equal to some constant $K$ on $\mathscr{X} \bigcap \mathscr{S}$ and to
$(c-1)\sigma_N^2/h$ on $\mathscr{S} \setminus \mathscr{X}$.

Transmitting $J_M(h)$, satisfying the power constraint with equality,
such that the transmitter power required for reliable communication is $P_M(h)=K,\forall h\in [h_1^*,h_2^*]$, and
$P_M(h)=(c-1)\sigma_N^2/h, \forall h\in [0,\infty)\setminus (h_1^*,h_2^*]$, for some $h_1^*< h_2^* \in \mathbb{R}_+$ and
some constant $K \in \mathbb{R}_+\bigcup \{\infty\}$ is an optimal jammer strategy for the maximin problem.
(Note that $P_M(h)$ should be continuous at $h_1^*$.)

The values $K,~h_1^*$ and $h_2^*$ that maximize the outage probability can be found by solving the following problem:

\begin{gather}
\textrm{Find $\min_{K} \int_{h_0}^{\infty} f(h)dh$, where} \nonumber\\
\textrm{$h_0$ is given by $\int_{h_0}^{h_2}K f(h)dh +\int_{h_2}^{\infty}\frac{c-1}{h}\sigma_N^2 f(h)dh  = \mathcal{P}$,}\\
\textrm{$h_1$ is given by $h_1=\frac{c-1}{K}\sigma_N^2$,}\\
\textrm{and $h_2$ is given by $\int_{h_1}^{h_2}\left( \frac{hK}{c-1}-\sigma_N^2 \right) f(h)dh  = \mathcal{J}$}.
\end{gather}
\end{prop}
$\blacksquare$
\vspace*{4pt}

The power allocation is depicted in Figure \ref{figure1}. The convex decreasing curve represents the original
required transmitter power, without the presence of a jammer and satisfies the equation $P_M=(c-1)\sigma_N^2/h$.
Notice how by picking some $K$, we can determine $h_1$, $h_2$ and $h_0$ (in this order), and then find the probability of
outage as $P_{out}(h_1)=1-\mathfrak{m}[(h_0,\infty)]$.
The optimal $K$, resulting in $h_1^*$, $h_2^*$ and $h_0^*$, is the one minimizing the $\mathfrak{m}$-measure
of the set $(h_0,\infty)$.
\begin{figure}[h]
\centering
\includegraphics[scale=1.0]{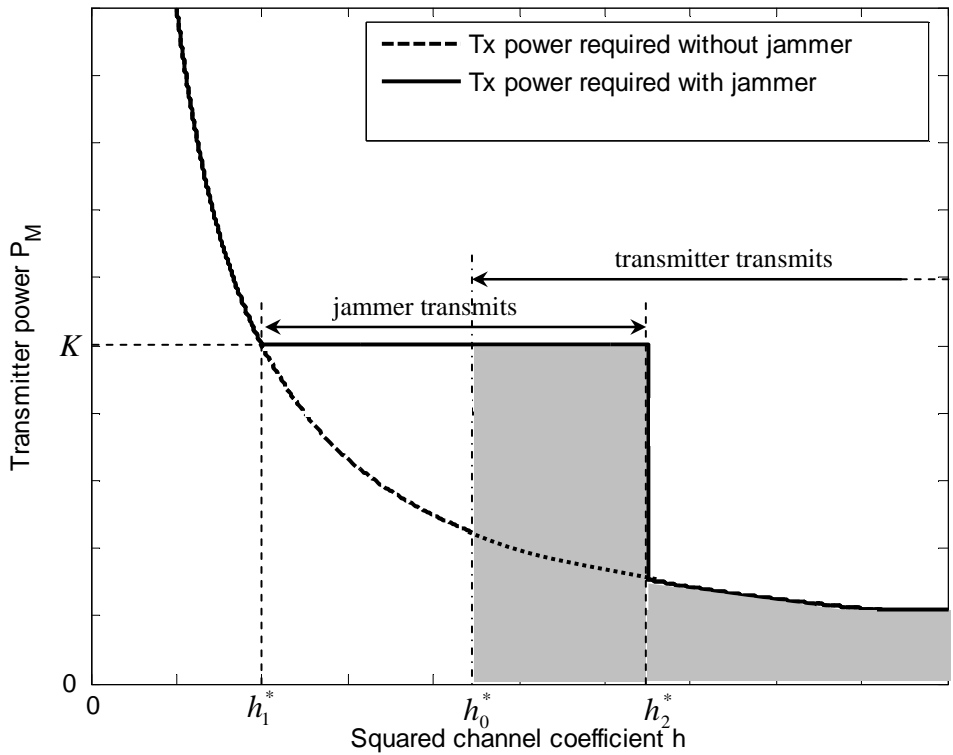}
\caption{Maximin solution for $M=1$ - power distribution between frames}\label{figure1}
\end{figure}

For the minimax scenario the jammer will not transmit any power over a frame if outage is not going to be induced or if
the transmitter is not present, i.e. $\mathscr{X}\subset \mathscr{S}$.
The jammer will start allocating power to the frames over which an outage is easiest to induce,
and go on with this technique until the average power reaches the limit set by its power constraint.
Obviously, the jammer prefers the frames for which the required $J_M(h)$ is less.
The optimal transmitter's strategy is to allocate its power such that the required
$J_M(h)$ is constant on the whole set $\mathscr{S}$, and hence to display no weakness.

These considerations are formalized in Proposition \ref{thm1_long_term} below.

\vspace*{4pt}
\begin{prop}\label{thm1_long_term}
For the minimax scenario, the transmitter's optimal way to allocate its power is to make the required jamming power remain 
equal to some constant $K$ on all of $\mathscr{X}$.
Transmitting $P_M(h)$, satisfying the power constraint with equality,
such that the required $J_M(h)$ equals $K$ for $h\in [h_x^*,\infty)$, and $J_M(h)=0~\forall h\in [0,h_x^*)$,
for some $h_x^* \in \mathbb{R}_+$, is an optimal transmitter strategy for the minimax problem.
The values $K$ and $h_x^*$ that minimize the outage probability can be found by solving the following problem numerically:

\begin{gather}
\textrm{Find $\max_{K} \int_{h_0}^{\infty} f(h)dh$, where} \nonumber\\
\textrm{$h_0$ is given by $\int_{h_x}^{h_0} Kf(h)dh = \mathcal{J}$,}\\
\textrm{$h_x$ is given by $\int_{h_x}^{\infty} \frac{(c-1)(K+\sigma_N^2)}{h}f(h)dh = \mathcal{P}$}. 
\end{gather}
\end{prop}
$\blacksquare$
\vspace*{4pt}

The numerical problem is described in Figure \ref{figure2}.
Notice how by picking some $K$, we can determine $h_x$ and $h_0$ (in this order), and then find the probability of
outage as $P_{out}(h_1)=1-\mathfrak{m}[(h_0,\infty)]$.
The optimal $K$, resulting in $h_x^*$  and $h_0^*$, is the one maximizing the $\mathfrak{m}$-measure of the set $(h_0,\infty)$.
Note that the jammer does not necessarily have to jam on an interval of the form $[h_x,h_0]$. The jammer's choice space
(the set of frames out of which the jammer picks its set $\mathscr{X}$) is an indifferent one, i.e. the jammer can randomly pick
$\mathscr{X}\subset [h_x,\infty)$ as long as its measure satisfies $K\mathfrak{m}(\mathscr{X})=\mathcal{J}$. However, for the
purpose of computing the outage probability, the representation of $\mathscr{X}$ as an interval is convenient and incurs no
loss of generality.

\begin{figure}[h]
\centering
\includegraphics[scale=1.0]{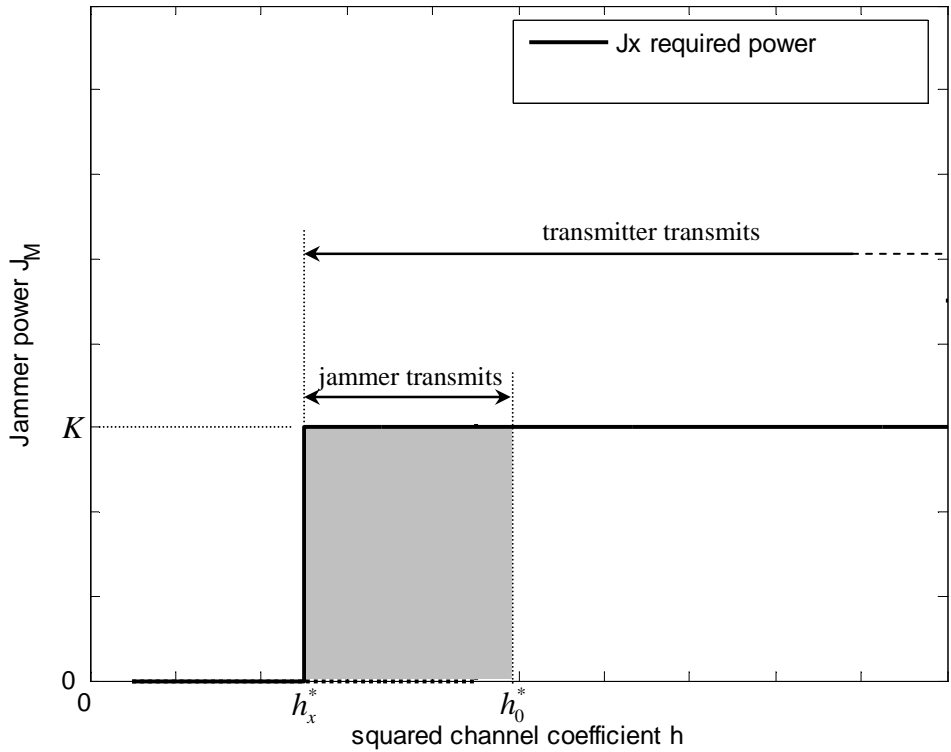}
\caption{Minimax solution for $M=1$ - power distribution between frames}\label{figure2}
\end{figure}


\subsection{Numerical Results}\label{ss5}

We have computed
the outage probabilities for both minimax and maximin problems when
$M=1$ and $M=2$.  The channel coefficients are assumed i.i.d. exponentially
distributed with parameter $\lambda=1/6$.  Figure~\ref{simulation1} 
shows the outage probability vs. the maximum allowable average transmitter
power $\mathcal{P}$ for fixed $\mathcal{J}=10$ when $R=1$.

For comparison purposes, we also plotted the results for the case when $M=\infty$, which are
readily available from Part I of this paper \cite{myself3}.

\begin{figure}
\centering
\includegraphics[scale=1.0]{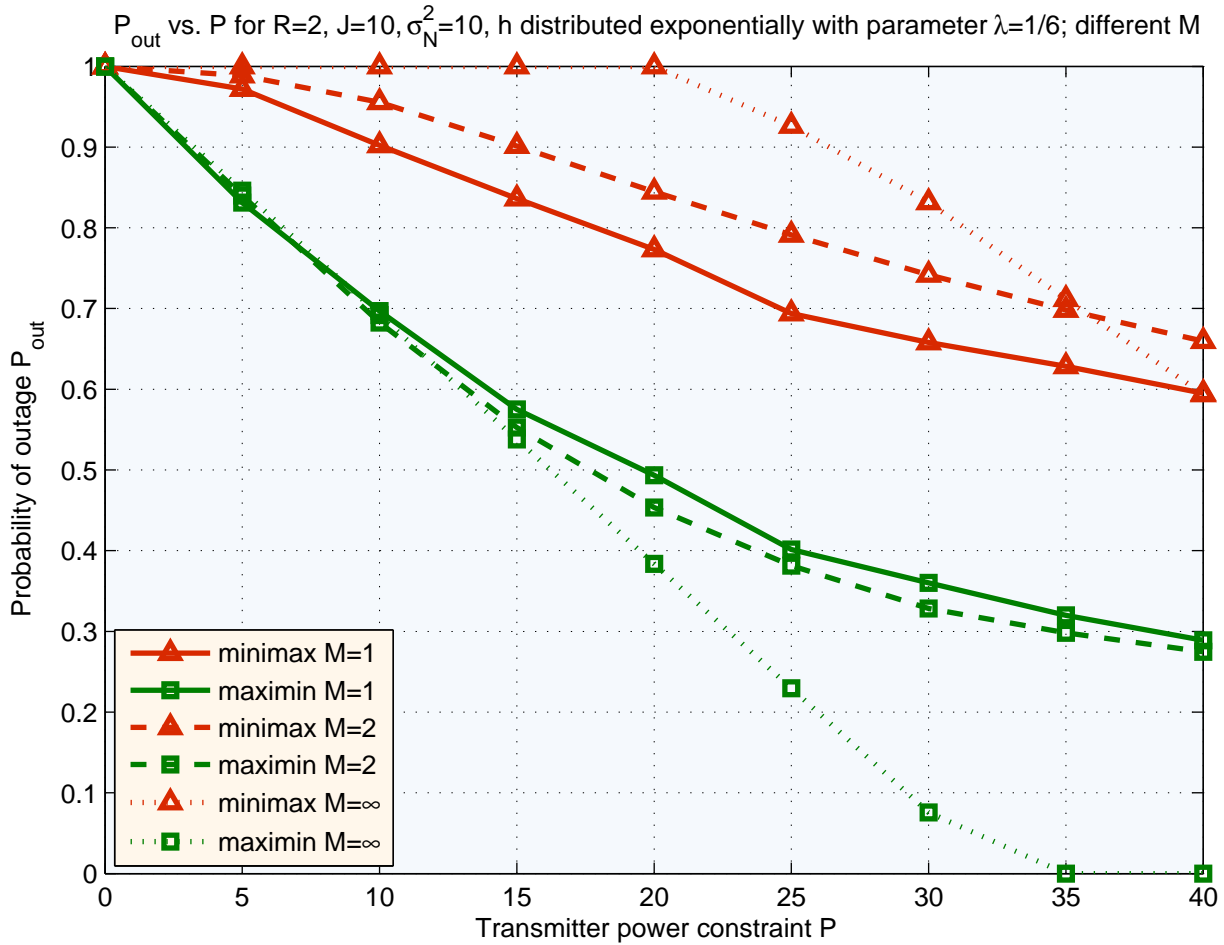}
\caption{Outage probability vs. $\mathcal{P}$ for $M=1$ and $M=2$ and $M\to\infty$ when $\mathcal{J}=10$ -- minimax and maximin cases.
We take $R=2$.}\label{simulation1}
\end{figure}

Numerical results demonstrate a sharp difference between the minimax solutions and the maxmin solutions, which
demonstrates the non-existence of Nash-equilibria of pure strategies for our two-person zero-sum game with full CSI.

Note the behavior of the outage probability when the number of blocks per frame $M$ is increased.
At low transmitter powers, the increase of $M$ produces an increase in the outage probability for both the minimax,
and the maximin scenarios.

On the contrary, at higher transmitter powers a lower outage probability is obtained for both the minimax
and the maximin cases when $M$ is larger. This behavior can be summarized as follows: the more powerful player
will use the available diversity to its own advantage.


\section{CSI Available to All Parties. Jamming Game with Long-Term Power Constraints: Mixed Strategies}\label{section3}

We have already seen that the maximin and minimax solutions of the jamming game
when only pure strategies are allowed do not agree, and thus our game has no
Nash equilibrium of pure strategies. However, recall that the solution of the minimax problem with pure strategies
can often be a good characterization of a practical jamming situation (e.g. when the jammer does
not transmit unless it senses that the transmitter is on) and can always serve as a lower bound
on the system's performance.

This aside, a Nash equilibrium is still the preferred characterization of jamming games, and
since such an equilibrium exists for our problem only when mixed strategies are allowed, the current section is dedicated
to the derivation of such a saddlepoint.

Unlike the fast fading scenario of \cite{myself3}, the frames in our slow-fading parallel-channels model
are not equivalent. Each frame is characterized by a different realization of the channel vector $\mathbf{h}$.
This is why our present scenario is even more involved than the one in \cite{myself3}, and requires three levels
of power control instead of two.

As before, our approach to the problem is a contradictory one. We study the power control levels starting with the
``finest'' one, and show that if our conditions for power allocations are not satisfied, then the strategy is suboptimal.
The reason why an additional (third) level of power control appears here is a combination of the facts that
we study mixed strategies and the frames are not all equivalent as in \cite{myself3}.
Namely, to cover all possible probabilistic strategies, we need to dedicate a level of power control to the
power allocation between frames with the same channel realizations (i.e. equivalent frames) and an additional
level of power control for the power allocation between frames with different channel realizations. 
Along with the power allocation within frames, these problems cover all possible cases.

\subsection{Power allocation within a frame}\label{ss3_1}

The third level of power control deals with the optimal power allocation between the blocks in a frame, once
the transmitter is given the channel vector $\mathbf{h}$ characterizing the frame and allocated power $P_M$, and
the jammer is given the channel vector and its allocated power $J_M$.

At this point, the third level of power control resembles the two-player, zero-sum game of (\ref{game21})
and (\ref{game22}) having the mutual information calculated over a frame $I_M$ as cost function.
However, none of the players knows the other player's constraints, because $(P_M,J_M)$ is a random event.
Theorem \ref{thm1_long_term_mixed_sf} below provides the optimal transmitter/jammer strategies for
power allocation within a frame.

\vspace*{4pt}
\begin{thm}\label{thm1_long_term_mixed_sf}
Given a frame with channel vector $\mathbf{h}$ and a realization $(p_M,j_M)$ of $(P_M,J_M)$,
let $\mathscr{P}_M(j_M)$ denote the solution of \emph{Problem 1}
in Section \ref{section2} with $J_M=j_M$, and $\mathscr{J}_M(p_M)$ denote the solution of \emph{Problem 2}
in Section \ref{section2} with $P_M=p_M$.

The transmitter's optimal strategy is the solution of the game
in (\ref{game21}) and (\ref{game22}), where the jammer is constrained to $\frac{1}{M}\sum_{m=1}^{M-1}J_m \leq \mathscr{J}_M(p_M)$ and
the transmitter is constrained to $\frac{1}{M}\sum_{m=1}^{M-1}P_m \leq p_M$.
The jammer's optimal strategy is the solution of the game in (\ref{game21}) and (\ref{game22}), where the
transmitter is constrained to $\frac{1}{M}\sum_{m=1}^{M-1}P_m \leq \mathscr{P}_M(j_M)$ and the jammer is constrained to
$\frac{1}{M}\sum_{m=1}^{M-1}J_m \leq j_M$.
\end{thm}
\vspace*{4pt}
\begin{proof}
The proof is very similar to the proof of Theorem 5 of \cite{myself3} and is deferred to
Appendix \ref{app3}
\end{proof}


\subsection{Power allocation between frames with the same channel vector}\label{ss3_3}

Due to the form of the optimal second level power allocation strategies described in the previous subsection,
the probability that a given frame is in outage can be expressed as
\be\label{alternpoutexpr}
P_{out,\mathbf{h}}=Pr\{J_M\geq \mathscr{J}_M(P_M)\}={}\nonumber\\
{}=1-Pr\{P_M\geq \mathscr{P}_M(J_M)\},
\ee
where $P_M(j_M)$ is the strictly increasing, unbounded and concave function (see Proposition \ref{propconcave})
that characterizes the frame. Note that a pair of strategies can only be optimal if $P_{out,\mathbf{h}}$ above
is the Nash equilibrium of a jamming game played over the frames characterized by the same channel vector $\mathbf{h}$.
This means that if the transmitter and jammer decide to allocate powers $P_{M,\mathbf{h}}$ and $J_{M,\mathbf{h}}$ respectively
to frames with channel vector $\mathbf{h}$, they should not allocate the same amount of power to each of these frames.
Instead, they should use power levels given by the realizations of two random variables $P_M$ and $J_M$ with
distribution functions $\left(F_P(p_M),F_J(j_M)\right)$ given in the following theorem.

\vspace*{4pt}
\begin{thm}\label{thm2_long_term_mixed}
The unique Nash equilibrium of mixed strategies of the two-player, zero-sum game
with average power constraints described by
\be\label{game_ss3_2}
\min_{F_P(p_M):\expec_{F_P}P_M\leq P_M(\mathbf{h})}\max_{F_J(j_M):\expec_{F_J}J_M\leq J_M(\mathbf{h})}P_{out,\mathbf{h}},
\ee
where $\expec_{F_P}$ and $\expec_{F_J}$ denote expectations with respect to the distributions $F_P(p_M)$ and $F_J(j_M)$,
is attained by the pair of strategies $\left(F_P(p_M),F_J(j_M)\right)$ satisfying:
\be\label{gp1_01}
F_P(\mathscr{P}_M(y))\sim k_p\mathbb{U}([0,2v])(y)+(1-k_p)\Delta_0(y),
\ee 
\be\label{gp1_02}
F_J(\mathscr{J}_M(x))\sim k_j\mathbb{U}([0,J_M(2v)])(x)+(1-k_j)\Delta_0(x),
\ee 
where $\mathbb{U}([r,t])(\cdot)$ denotes the CDF of a uniform distribution
over the interval $[r,t]$, and $\Delta_0(\cdot)$ denotes the CDF of a Dirac distribution (i.e. a step function),
and the parameters $k_p,k_j\in [0,1]$ and $v\in [\max \{J_{M,\mathbf{h}}, \mathscr{J}_M(P_{M,\mathbf{h}})/2 \}, \infty)$ are uniquely
determined from the following steps:
\begin{enumerate}
\item Find the unique value $v_0$ which satisfies:
\be
P_{M,\mathbf{h}}J_{M,\mathbf{h}}=[\mathscr{P}_M(2v_0)-P_{M,\mathbf{h}}](2v_0-J_{M,\mathbf{h}}).
\ee
\item Compute $S(v_0)=\int_{0}^{2v_0}\mathscr{P}_M(y)dy-2v_0P_{M,\mathbf{h}}$.
\item If $S(v_0)<0$, then $v$ is the unique solution of 
\be
\int_{0}^{2v}\mathscr{P}_M(y)dy-2vP_{M,\mathbf{h}}=0,
\ee
\be
k_p=1
\ee
and 
\be
k_j=\frac{J_{M,\mathbf{h}}\mathscr{P}_M(2v)}{2v[\mathscr{P}_M(2v)-P_{M,\mathbf{h}}]}.
\ee
\item If $S(v_0)=0$ then $v=v_0$, $k_p=k_j=1$.
\item If $S(v_0)>0$, then $v$ is the unique solution of 
\be
\int_{0}^{2v}\mathscr{P}_M(y)dy-\mathscr{P}_M(2v)(2v-J_{M,\mathbf{h}})=0,
\ee
\be
k_p=\frac{2vP_{M,\mathbf{h}}}{\mathscr{P}_M(2v)[2v-J_{M,\mathbf{h}}]}
\ee
and 
\be
k_j=1.
\ee
\end{enumerate}
\end{thm}
\vspace*{4pt}
\begin{proof}
The proof follows directly from Theorem 9  in Appendix III of \cite{myself3},
by substituting $x=P_M$, $y=J_M$, $g(y)=\mathscr{P}_M(y)$, $g^{-1}(x)=\mathscr{J}_M(x)$,
$a=P_{M,\mathbf{h}}$ and $b=J_{M,\mathbf{h}}$. It is also interesting to note that the condition
$\int_0^{b}g(y)dy<\int_{g(b)}^{\infty}g^{-1}(x)dx$ is satisfied because $\mathscr{P}_M(y)$
is unbounded.
\end{proof}
\vspace*{4pt}

\vspace*{4pt}
{\bf \emph{Particular case: $M=1$}} \vspace*{4pt}

For $M=1$ the first (intra-frame) level of power control is inexistent. For a given channel realization
$h$ we can readily derive the \emph{affine} function $P_M(j_M)$
in (\ref{alternpoutexpr}) as
\be
P_M(j_M)=\frac{c-1}{h}(j_M+\sigma_N^2)
\ee
where $c=\exp(R)$.
If we use the particularization of the general solution of Theorem \ref{thm2_long_term_mixed} to affine functions,
as in the last part of Appendix III of \cite{myself3}, we obtain the outage probability as

\be\label{poutmixM1_1}
P_{out,h}=1-\frac{\frac{hP_M(h)}{c-1}}{J_M(h)\left[1+\sqrt{1+2\frac{\sigma_N^2}{J_M(h)}}\right]+\sigma_N^2}\nonumber\\
\textrm{if}~~\frac{hP_M(h)}{c-1}\leq \frac{1}{2}J_M(h)\left[1+\sqrt{1+2\frac{\sigma_N^2}{J_M(h)}}\right]+\sigma_N^2,
\ee
and
\be\label{poutmixM1_2}
P_{out,h}=\frac{\frac{1}{2}J_M(h)}{\frac{hP_M(h)}{c-1}-\sigma_N^2}\nonumber\\
\textrm{if}~~\frac{hP_M(h)}{c-1}> \frac{1}{2}J_M(h)\left[1+\sqrt{1+2\frac{\sigma_N^2}{J_M(h)}}\right]+\sigma_N^2.
\ee

The transmitter and jammer strategies that achieve these payoffs are such that
\be\label{gp1_01_nocsi}
F_P(x)\sim k_p\mathbb{U}([\frac{c-1}{h} \sigma_N^2,2v\frac{c-1}{h}+\frac{c-1}{h} \sigma_N^2])(x)+\nonumber\\
+(1-k_p)\Delta_0(x),\nonumber
\ee 
\be\label{gp1_02_nocsi}
F_J(y)\sim \frac{2v}{2v+\sigma_N^2}k_j\mathbb{U}([0,2v])(y)+(1-\frac{2v}{2v+\sigma_N^2}k_j)\Delta_0(y).\nonumber
\ee 
The parameters $k_p,k_j\in [0,1]$ and $v\in [\max \{J_M(h), \mathscr{J'}_M(P_M(h))/2 \}, \infty)$ are uniquely
determined from the following steps:

\begin{enumerate}
\item If
\be
\frac{hP_M(h)}{c-1}\leq \frac{1}{2}J_M(h)\left[1+\sqrt{1+2\frac{\sigma_N^2}{J_M(h)}}\right]+\sigma_N^2,
\ee
then
\be
v=\frac{1}{2}J_M(h)\left[1+\sqrt{1+\frac{2\sigma_N^2}{J_M(h)}} \right],
\ee
\be
k_p=\frac{2vP_M(h)}{\frac{c-1}{h}(2v+\sigma_N^2)(2v-J_M(h))}
\ee
and 
\be
k_j=1.
\ee

\item If 
\be
\frac{hP_M(h)}{c-1}> \frac{1}{2}J_M(h)\left[1+\sqrt{1+2\frac{\sigma_N^2}{J_M(h)}}\right]+\sigma_N^2,
\ee
then 
\be
v=\frac{P_M(h)-\frac{c-1}{h}\sigma_N^2}{\frac{c-1}{h}},
\ee
\be
k_p=1
\ee
and 
\be
k_j=\frac{\frac{c-1}{h}J_M(h)(2P_M(h)-\frac{c-1}{h}\sigma_N^2) }{2(P_M(h)-\frac{c-1}{h}\sigma_N^2)^2}.
\ee

\end{enumerate}

The special form of this solution will be used in the next subsection to derive the overall Nash equilibrium of the
mixed strategies game for $M=1$.


\subsection{Power allocation between frames with different channel vectors}\label{ss3_4}

In the previous subsections we have described the optimal power control strategies for
given particular channel realization $\mathbf{h}$, and  transmitter and jammer power levels $P_{M,\mathbf{h}}$ and
$J_{M,\mathbf{h}}$ respectively.
The first level of power control,which is the subject of this subsection, deals with allocating the powers specified by
the transmitter and jammer average power constraints $\mathcal{P}$ and $\mathcal{J}$ between different channel
vectors.
In other words, we are now concerned with solving the problem
\be\label{probl_1_ss3_2}
\min_{P_M(\mathbf{h}): \expec_{\mathbf{h}}P_M(\mathbf{h})\leq\mathcal{P}}
\max_{J_M(\mathbf{h}): \expec_{\mathbf{h}}J_M(\mathbf{h})\leq\mathcal{J}} \expec_{\mathbf{h}}[\nonumber\\
P_{out,\mathbf{h},P_M(\mathbf{h}),J_M(\mathbf{h})}]
\ee
where $P_{out,\mathbf{h},P_M(\mathbf{h}),J_M(\mathbf{h})}$ (also denoted as $P_{out,\mathbf{h}}$) is the outage probability of a
frame characterized by the channel vector $\mathbf{h}$ and to which the transmitter allocates power
$P_M(\mathbf{h})$, and the jammer allocates power $J_M(\mathbf{h})$.
Note that $P_{out,\mathbf{h},P_M(\mathbf{h}),J_M(\mathbf{h})}$ can be easily computed according to the second and
third levels of power control already presented.

However, the Nash equilibrium of the game in (\ref{probl_1_ss3_2}) above is highly dependent on the
result of the second level of power control. Since finding a closed form solution for
the second level is still an open problem, a general solution for the first level of power control
is not available at this time. 

However, we next provide a Nash equilibrium for the particular case when $M=1$.

\vspace*{4pt}
{\bf \emph{Particular case: $M=1$}} \vspace*{4pt}

We start by pointing out the following important property of the second-level power control strategies for $M=1$.

\begin{prop}\label{convconcM1}
The outage probability $P_{out,h}$ given in (\ref{poutmixM1_1}) and (\ref{poutmixM1_2}) above is a continuous
function of both arguments.
Moreover, $P_{out,h}$ is a strictly decreasing, convex function of $P_M(h)$ for fixed $J_M(h)$ and 
a strictly increasing, concave function of $J_M(h)$ for fixed $P_M(h)$.
\end{prop}
\begin{proof}
In the remainder of this section we shall denote the case when $\frac{hP_M(h)}{c-1}\leq \frac{1}{2}J_M(h)
\left[1+\sqrt{1+2\frac{\sigma_N^2}{J_M(h)}}\right]+\sigma_N^2$ by \emph{Case 1} and the case when
$\frac{hP_M(h)}{c-1}> \frac{1}{2}J_M(h)\left[1+\sqrt{1+2\frac{\sigma_N^2}{J_M(h)}}\right]+\sigma_N^2$ by
\emph{Case 2}.

It is straightforward to check that when $\frac{hP_M(h)}{c-1}= \frac{1}{2}J_M(h)
\left[1+\sqrt{1+2\frac{\sigma_N^2}{J_M(h)}}\right]+\sigma_N^2$ we get $P_{out,h}=\frac{1}{1+\sqrt{1+2\frac{\sigma_N^2}{J_M(h)}}}$
by using either of the relations in (\ref{poutmixM1_1}) or (\ref{poutmixM1_2}). Thus, the continuity of $P_{out,h}$ follows
immediately.

If we evaluate the derivatives for \emph{Case 1}
\be
\frac{dP_{out,h}}{dP_M(h)}=-\frac{\frac{h}{c-1}}{J_M(h)\left[1+\sqrt{1+2\frac{\sigma_N^2}{J_M(h)}}\right]+\sigma_N^2}
\ee
and for \emph{Case 2}
\be
\frac{dP_{out,h}}{dP_M(h)}=-\frac{\frac{c-1}{h}J_M(h)}{2(P_M(h)-\frac{c-1}{h}\sigma_N^2)^2}
\ee
we note that when $J_M(h)$ is fixed, $P_{out,h}$ is a strictly decreasing function of $P_M(h)$,
affine in \emph{Case 1} and strictly convex in \emph{Case 2}. Moreover, $\frac{dP_{out,h}}{dP_M(h)}$ is continuous, which
makes $P_{out,h}$ an overall strictly decreasing, convex function of $P_M(h)$.

Similar (but symmetric) properties hold for the derivatives
\be
\frac{dP_{out,h}}{dJ_M(h)}=\frac{\frac{h}{c-1}}{J_M(h)\left[1+\sqrt{1+2\frac{\sigma_N^2}{J_M(h)}}\right]+\sigma_N^2}\cdot\nonumber\\
\cdot \frac{P_M(h)}{J_M(h)\sqrt{1+2\frac{\sigma_N^2}{J_M(h)}}},
\ee
for \emph{Case 1} and
\be
\frac{dP_{out,h}}{dJ_M(h)}=\frac{1}{2}\frac{1}{\frac{h}{c-1}P_M(h)-\sigma_N^2}
\ee
for \emph{Case 2},
yielding $P_{out,h}$ an overall strictly increasing, concave function of $J_M(h)$ (strictly concave in \emph{Case 1}
and affine in \emph{Case 2}).
\end{proof}

The result of Proposition \ref{convconcM1} implies that the overall outage probability $\expec_{h}P_{out,h}$ is a
convex function of $\{P_M(h)\}$ for fixed $\{J_M(h)\}$ and a concave function of $\{J_M(h)\}$ for fixed $\{P_M(h)\}$.
Since the set of strategies $\{P_M(h), J_M(h)\}$ is convex, there always exists a saddlepoint of the
game in (\ref{probl_1_ss3_2}) \cite{aubin}.
The importance of this result should be noted, since it implies that a Nash equilibrium of mixed strategies
of the two-person, zero-sum game in (\ref{probl_1_ss3_2}) can be achieved by only looking for pure strategies.
Recall that any Nash equilibrium of pure strategies is also a Nash equilibrium of mixed strategies, and that
for a two-person, zero-sum game all Nash equilibria share the same value of the cost function \cite{meyerson}.

Any saddlepoint of (\ref{probl_1_ss3_2}) has to satisfy the KKT conditions associated with the maximization
and minimization problems of (\ref{probl_1_ss3_2}) simultaneously. The next Proposition shows these
KKT conditions are not only necessary, but also sufficient for determining a saddlepoint.
The proof is deferred to Appendix \ref{app3}.

\begin{prop}\label{prop_saddlepoints_kkt}
For our two-player, zero-sum game of (\ref{probl_1_ss3_2}),
any solution of the joint system of KKT conditions associated with the maximization and minimization problems
yields a Nash equilibrium.
\end{prop}

We can now solve the KKT conditions associated with the maximization and minimization problems
of (\ref{probl_1_ss3_2}) simultaneously.
For \emph{Case 1}, these are
\be\label{KKT_case1_1}
-\frac{\frac{h}{c-1}}{J_M(h)\left[1+\sqrt{1+2\frac{\sigma_N^2}{J_M(h)}}\right]+\sigma_N^2}+\lambda-\gamma(h)=0
\ee
and
\be\label{KKT_case1_2}
-\frac{\frac{h}{c-1}}{J_M(h)\left[1+\sqrt{1+2\frac{\sigma_N^2}{J_M(h)}}\right]+\sigma_N^2}\cdot\nonumber\\
\cdot\frac{P_M(h)}{J_M(h)\sqrt{1+2\frac{\sigma_N^2}{J_M(h)}}}+\mu-\delta(h)=0,
\ee
where $\gamma(h)$ and $\delta(h)$ are the complementary slackness conditions satisfying $\gamma(h)P_M(h)=0$
and $\delta(h)J_M(h)=0$, and where $\mu, ~\lambda ~\geq 0$.
From (\ref{KKT_case1_2}) we get
\be
P_M(h)=\frac{\mu}{\lambda}J_M(h)\sqrt{1+2\frac{\sigma_N^2}{J_M(h)}},
\ee
resulting in
\be
J_M(h)=\left[\sqrt{\left(\frac{\lambda}{\mu}\right)^2 P_M(h)^2+\sigma_N^4}-\sigma_N^2\right]_+,
\ee
which in combination with (\ref{KKT_case1_1}) yields
\be
P_M(h)=\left[\frac{h}{c-1}\frac{\mu}{2\lambda^2}-\frac{\mu(c-1)}{2h}\sigma_N^4\right]_+,
\ee
where we denote $[x]_+=\max\{ x,0 \}$.
Under this solution, the condition for being under \emph{Case 1},
\be
\frac{hP_M(h)}{c-1}\leq \frac{1}{2}J_M(h)\left[1+\sqrt{1+2\frac{\sigma_N^2}{J_M(h)}}\right]+\sigma_N^2
\ee
translates to
\be\label{cond_case1}
\frac{2\mu h}{\lambda (c-1)}\leq 1+\sqrt{1+4\sigma_N^2 \mu^2 \left(\sigma_N^2+\frac{1}{\mu}\right)}=\nonumber\\
=2(1+\sigma_N^2 \mu).
\ee

Note that $P_M(h)=0$ if and only if $J_M(h)=0$, and this happens when $h\leq h_{0/1}$, where
\be
h_{0/1}= \lambda(c-1)\sigma_N^2.
\ee

Writing the KKT conditions for \emph{Case 2} under the assumption that $P_M(h),~J_M(h)\geq 0$ we obtain
\be\label{KKT_case2_1}
-\frac{\frac{h}{c-1}J_M(h)}{2\left(\frac{h}{c-1}P_M(h)-\sigma_N^2\right)^2}+\lambda-\gamma(h)=0
\ee
and
\be\label{KKT_case2_2}
-\frac{1}{2\left(\frac{h}{c-1}P_M(h)-\sigma_N^2\right)}+\mu-\delta(h)=0,
\ee
which yield
\be
J_M(h)=\frac{c-1}{h}\frac{\lambda}{2\mu ^2}
\ee
and
\be
P_M(h)=\frac{c-1}{h}\left(\frac{1}{2\mu}+\sigma_N^2\right).
\ee
Note that in this case both $P_M(h)$ and $J_M(h)$ are strictly positive for finite $h$.
Under this solution, the condition for being under \emph{Case 2},
\be
\frac{hP_M(h)}{c-1}> \frac{1}{2}J_M(h)\left[1+\sqrt{1+2\frac{\sigma_N^2}{J_M(h)}}\right]+\sigma_N^2
\ee
translates to
\be\label{cond_case2}
\frac{2\mu h}{\lambda (c-1)}> 1+\sqrt{1+4\sigma_N^2 \mu^2 \frac{h}{\lambda(c-1)}}.
\ee

Forcing the right-hand side of (\ref{cond_case1}) to equal the right-hand side of (\ref{cond_case2})
we get the value of $h$ which is at the boundary between \emph{Case 1} and \emph{Case 2}:
\be
h_{1/2}= \lambda(c-1)(\frac{1}{\mu}+\sigma_N^2).
\ee

A close inspection of the expressions of $P_M(h)$ and $J_M(h)$ for the two cases shows that they are both increasing
functions of $h$ under \emph{Case 1} and decreasing functions of $h$ under \emph{Case 2}, and moreover, they are
both continuous in $h_{1/2}$.
To summarize the results above, the optimal transmitter/jammer first level power control strategies are given in
(\ref{mixed_firstlevel_tx}) and (\ref{mixed_firstlevel_jx}) below, respectively. The constants $\lambda$ and $\mu$ can
be obtained from the power constraints $\expec_h P_M(h)=\mathcal{P}$ and $\expec_h J_M(h)=\mathcal{J}$.

\be\label{mixed_firstlevel_tx}
P_M(h)=\left\{\begin{array}{lcr}
0 ,& \textrm{if} & h\leq h_{0/1}\\
\frac{h}{c-1}\frac{\mu}{2\lambda^2}-\frac{\mu(c-1)}{2h}\sigma_N^2,& \textrm{if}  & h_{0/1}<h\leq h_{1/2}\\
\frac{c-1}{h}\left(\frac{1}{2\mu}+\sigma_N^2\right),& \textrm{if} & h> h_{1/2}
\end{array} \right.
\ee

\be\label{mixed_firstlevel_jx}
J_M(h)=\left\{\begin{array}{lcr}
0 ,& \textrm{if} & h\leq h_{0/1}\\
\sqrt{\left(\frac{\lambda}{\mu}\right)^2 (\frac{h}{c-1}\frac{\mu}{2\lambda^2}-
\frac{\mu(c-1)}{2h}\sigma_N^2)^2+\sigma_N^4}-\sigma_N^2,& \textrm{if} & h_{0/1}<h\leq h_{1/2}\\
\frac{c-1}{h}\frac{\lambda}{2\mu ^2} ,& \textrm{if} & h> h_{1/2}
\end{array} \right.
\ee


\subsection{Numerical results}

Figure \ref{fig_numericalresults3_2} shows the probability of outage obtained under the mixed strategies
Nash equilibrium, versus the transmitter power constraint $\mathcal{P}$, when $M=1$, for a fixed
rate $R=1$, noise power $\sigma_N^2=10$, a jammer power constraint $\mathcal{J}=10$
and a channel coefficient distributed exponentially, with parameter $\lambda=1/6$.
The maximin and minimax solutions of the pure strategies game are shown for comparison.
 
\begin{figure}[h]
\centering
\includegraphics[scale=1.0]{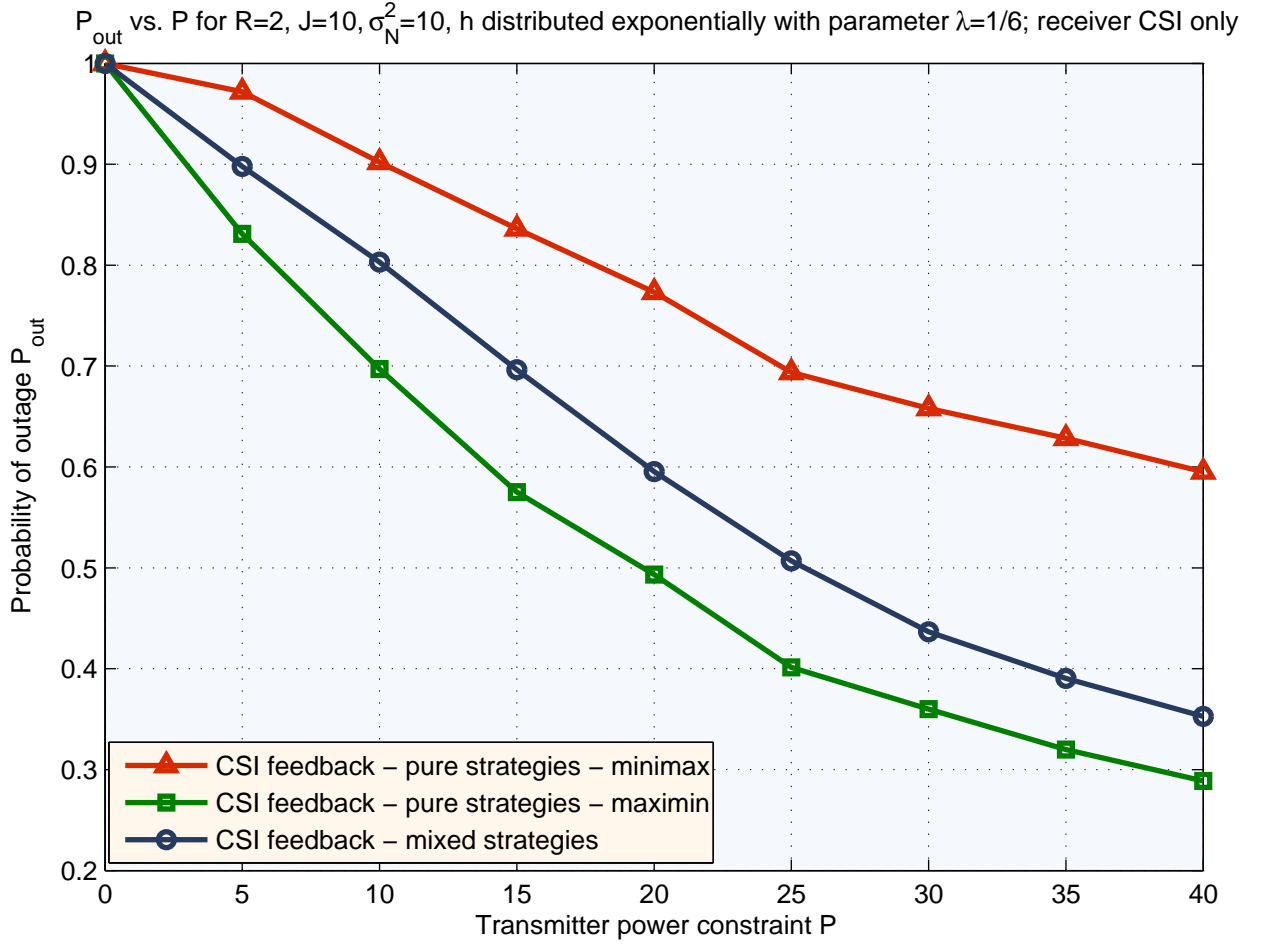}
\caption{Outage probability vs. transmitter power constraint $\mathcal{P}$ for $M=1$ when $\mathcal{J}=10$,
 $R=2$, $\sigma_N^2=10$ and $h$ is distributed exponentially, with
parameter $\lambda=1/6$.}\label{fig_numericalresults3_2}
\end{figure}

As expected, the solution of the of mixed strategies game
is better (from the transmitter's point of view) than the minimax and worse than the maximin solutions
of the pure strategies game.


\section{CSI Available Receiver Only. Jamming Game with Long-Term Power Constraints: Mixed Strategies}\label{section4}

In this section we investigate the scenario when the receiver does not feed back any channel state
information. Since we have already shown that the problem with long-term power constraints is the more
interesting and challenging one, and since the purpose of this section is to offer a comparison with previous results,
we further focus only on the scenario of average power constraints and mixed strategies.

Unlike in the corresponding Section V of \cite{myself3}, where all frames were equivalent because of the
fast fading channel, in our present scenario each frame is characterized by a particular channel realization. Since
this channel realization is not known to either the transmitter or the jammer, they both have to allocate
some power over each frame, in a random fashion, such that the transmitter minimizes and the jammer
maximizes the probability that the mutual information over the frame is less than the transmission rate $R$.
In its most general form, the game can be written as

\be\label{gamenocsi}
\min_{P_M:\expec P_M\leq \mathcal{P}}\max_{J_M:\expec J_M\leq \mathcal{J}}\expec_{P_M,J_M}\nonumber\\
\Bigg[\min_{P_m:\sum P_m\leq M P_M}\max_{J_m:\sum J_m\leq J_M}\nonumber\\
Pr\{ \sum_{m=0}^{M-1}\log\left(1+\frac{P_m h_m}{J_m+\sigma_N^2}\right)\leq MR\}
 \Bigg],
\ee
where $\expec_{P_M,J_M}$ denotes statistical expectation with respect to the probability distribution of
$P_M$ and $J_M$.
The form of (\ref{gamenocsi}) suggests two levels of power control: a first one which deals with the
allocation of power between different frames, and a second one which allocates the powers within each frame.

In solving the game, we start as before with the second level of power control. However, this level requires
an exact expression of $Pr\{ \sum_{m=0}^{M-1}\log\left(1+\frac{P_m h_m}{J_m+\sigma_N^2}\right)\leq MR\}$.
Note that this probability depends upon the probability distribution of the channel vector $\mathbf{h}$.
A practical way of solving the problem is the following.

Denote $S_m=\log\left(1+\frac{P_m h_m}{J_m+\sigma_N^2}\right)$ the random variable (depending on $h_m$)
which characterizes the instant mutual information over the $m$-th block of the frame.
We can write the cumulative distribution function (c.d.f.) of $S_m$ as
\be
F_{S_m}(x)=Pr\{S_m\leq x\}=\nonumber\\
=Pr\{h_m\leq (e^x-1)\frac{J_m+\sigma_N^2}{P_m}\}=F_{h}\left((e^x-1)\frac{J_i+\sigma_N^2}{P_i}\right)
\ee
where $F_h(x)$ is the c.d.f. of the channel coefficient $h_m$ and we assume that the channel
coefficients over all the blocks of a frame are independent and identically distributed random variables.

We can now compute the p.d.f. (assuming it exists) of $S_m$ as
\be
f_{S_m}(x)=\frac{dF_{S_m}(x)}{dx}
=\frac{dF_{h}\left((e^x-1)\frac{J_m+\sigma_N^2}{P_m}\right)}{dx}.
\ee 

Finally, our probability can be written as
\be
Pr\{ \sum_{m=0}^{M-1}\log\left(1+\frac{P_m h_m}{J_m+\sigma_N^2}\right)\leq MR\}=\nonumber\\
=\left(F_{S_0}*f_{S_1}*\ldots * f_{S_{M-1}}\right)(MR)
\ee
where $*$ denotes regular convolution.
Due to the intricate expression of this probability, as well as its dependence on the statistical
properties of the channel, we next focus exclusively on the simple case when $M=1$.

\vspace*{4pt}
{\bf \emph{Particular case: $M=1$}} \vspace*{4pt}

For $M=1$, we are only concerned with the first level of power control.
The game can be written as
\be
\min_{P_M:\expec P_M\leq \mathcal{P}}\max_{J_M:\expec J_M\leq \mathcal{J}}\nonumber\\
\expec_{P_M,J_M}Pr\{ P\leq (c-1)\frac{J_M+\sigma_N^2}{h}\}
\ee
or equivalently,
\be
\min_{P_M:\expec P_M\leq \mathcal{P}}\max_{J_M:\expec J_M\leq \mathcal{J}}\nonumber\\
\expec_{P_M,J_M}Pr\{ h\leq (c-1)\frac{J_M+\sigma_N^2}{P_M}\}.
\ee

In order to provide a good numerical comparison with the results of the previous sections,
assume that the channel coefficient $h$ has an exponential probability distribution
with parameter $\lambda$. Its cumulative distribution function can thus be written as
$F(h)=1-e^{-\lambda h}$, which enables us to write
\be
Pr\{ h\leq (c-1)\frac{J_M+\sigma_N^2}{P_M}\}=\nonumber\\
=1-\exp\left[-\lambda(c-1)\frac{J_M+\sigma_N^2}{P_M} \right].
\ee
Denote $H(P_M,J_M)=1-\exp\left[-\lambda(c-1)\frac{J_M+\sigma_N^2}{P_M} \right]$.

By computing the derivatives
\be
\frac{dH}{dP_M}=\nonumber\\
=-\lambda(c-1)\frac{J_M+\sigma_N^2}{P_M^2}\exp\left[-\lambda(c-1)\frac{J_M+\sigma_N^2}{P_M} \right]<0,
\ee

\be
\frac{d^2H}{dP_M^2}=\nonumber\\
=\lambda(c-1)\frac{J_M+\sigma_N^2}{P_M^3}\left[\lambda(c-1)\frac{J_M+\sigma_N^2}{P_M}+2 \right]\nonumber\\
\exp\left[-\lambda(c-1)\frac{J_M+\sigma_N^2}{P_M} \right]>0,
\ee

\be
\frac{dH}{dJ_M}=\nonumber\\
=\lambda(c-1)\frac{1}{P_M}\exp\left[-\lambda(c-1)\frac{J_M+\sigma_N^2}{P_M} \right]>0,
\ee
and
\be
\frac{d^2H}{dJ_M^2}=\nonumber\\
=-(\lambda(c-1)\frac{1}{P_M})^2\exp\left[-\lambda(c-1)\frac{J_M+\sigma_N^2}{P_M} \right]<0,
\ee
we notice that $H$ is a strictly decreasing, convex function of $P_M$ for a fixed $J_M$, and a
strictly increasing, concave function of $J_M$ for a fixed $P_M$.
Hence, a Nash equilibrium is achieved by uniformly distributing the transmitter's and jammer's powers between
the frames:
\be
\expec_{P_M:\expec P_M\leq \mathcal{P}}\left\{ 1-\exp\left[-\lambda(c-1)\frac{\mathcal{J}+\sigma_N^2}{P_M} \right] \right\}
\leq{}\nonumber\\
{}\leq 1-\exp\left[-\lambda(c-1)\frac{\mathcal{J}+\sigma_N^2}{\mathcal{P}} \right] \leq{}\nonumber\\
{}\leq \expec_{J_M:\expec J_M\leq \mathcal{J}}\left\{ 1-\exp\left[-\lambda(c-1)\frac{J_M+\sigma_N^2}{\mathcal{P}} \right] \right\}
\ee
This saddlepoint is an equilibrium of pure strategies, and hence also an equilibrium of mixed strategies.
Note that the existence of such an equilibrium of pure strategies might no longer hold for different
probability distributions of $h$, and this would demand a search for purely probabilistic strategies.
For example, when the c.d.f. of the channel coefficient $F(h)$ is not concave, then $F((c-1)\frac{J_M+\sigma_N^2}{P_M})$
is no longer a concave function of $J_M$, and hence the optimal jammer strategy is not deterministic.

Numerical evaluations of the system's performance under the present scenario are presented in the next subsection.


\subsection{Numerical results}\label{ss4_2}

The probability of outage as a function of the transmitter's power constraint $\mathcal{P}$
is shown in Figure \ref{fig_numericalresults5_2} for $M=1$, and under the assumption that both the transmitter and the jammer
distribute their powers uniformly over the frames.

For comparison, the maximin and minimax solutions of the pure strategies game and the mixed strategies Nash equilibrium, all
under the scenario that channel state information is fed back by the receiver, are also shown in the figure.

\begin{figure}[h]
\centering
\includegraphics[scale=1.0]{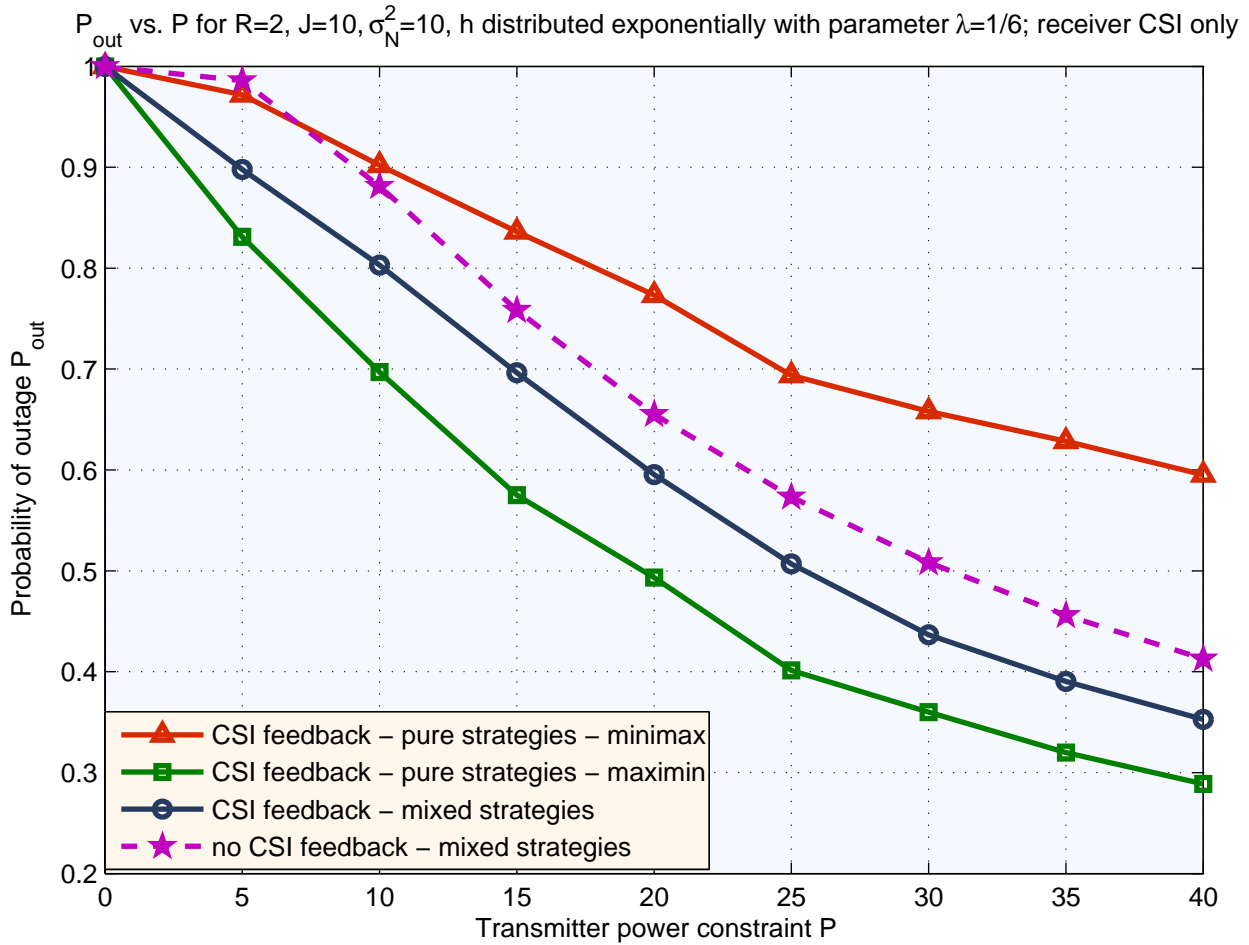}
\caption{Outage probability vs. transmitter power constraint $\mathcal{P}$ for $M=1$, with and without CSI feedback when $\mathcal{J}=10$,
 $R=2$, $\sigma_N^2=10$ and $h$ is distributed exponentially, with
parameter $\lambda=1/6$. (Mixed strategies.)}\label{fig_numericalresults5_2}
\end{figure}

Note that when the receiver does not feed back the CSI, the system performance suffers degradation. Unlike in the
fast fading scenario of \cite{myself3}, in the present slow fading scenario the increase in the outage probability is significant.
The difference is most visible at low transmitter powers, when not feeding back the channel state information amounts to
worse performance than the pessimistic (minimax) scenario with full CSI.


\section{Conclusions}\label{section5}

We have studied the jamming game over slow fading channels, with the outage probability as objective.
Similarly to the fast fading scenario, the game with full CSI and average (or long term) power constraints
does not have a Nash equilibrium of pure strategies. Nevertheless, we derived the minimax and maximin solutions
of pure strategies, which provide lower and upper bounds on the system performance, respectively.

In addition, we investigated the Nash equilibrium of mixed strategies. Compared to the fast fading scenario
\cite{myself3}, the Nash equilibrium for the slow fading, full CSI game is much more involved. The difference
comes from the fact that frames are not equivalent. In fact, instead of being characterized by the channel statistics
as in \cite{myself3}, the frames are now characterized by different channel realizations. This results in
the existence of an additional third level of power control.

We also showed that for parallel slow fading channels, the CSI feedback helps in the battle
against jamming, since if the receiver does not feed back the channel state information, the system's performance suffers
a significant degradation. We expect this degradation to decrease as the number of parallel channels $M$ increases,
until it becomes marginal for $M\to \infty$ (which can be considered as the case in \cite{myself3}).

These results, along with our conclusions from the first part of this paper \cite{myself3}, reveal an interesting
duality between the ways that different communication models behave with and without jamming.
As remarked in \cite{myself3}, under a fast fading channel with jamming, the feedback of channel state information brings little
benefits in terms of the overall probability of outage. The same tendency is observed for the fast fading channel without
jamming in \cite{goldsmith} (although the performance measure therein is the ergodic capacity).
However, \cite{caire} shows that for a parallel slow fading channel, the CSI feedback is quite important.
The improvement of the probability of outage when the channel coefficients are perfectly known to the transmitter is no
longer negligible. The results of our present paper demonstrate that even in the presence of a jammer (which can eavesdrop the
feedback channel and hence obtain the same CSI as the transmitter), CSI feedback improves the transmission considerably.


\appendices
\section{Short-Term Power Constraints - Proofs of Main Results}\label{app1}

\subsection{Proof of Proposition \ref{prop_short_term}}

The proof is an adaptation of the results in Section IV.B of \cite{uluk1}, regarding uncorrelated jamming with CSI at the transmitter.
The only difference is that in our case, the power constraints and cost function involve short-term, temporal averages,
while in \cite{uluk1}, they are expressed in terms of statistical averages. 
Nevertheless, the same techniques can be applied.

The set of all pairs $(P(h),J(h))$ satisfying the power constraints is convex, since the power constraints are linear functions
of $P(h)$ and $J(h)$, respectively.
Moreover, the cost function 
\begin{gather}
I_M(\mathbf{h}, P(h),J(h))=\frac{1}{M}\sum_{m=0}^{M-1} \log (1+\frac{h_m P_m}{\sigma_N^2+J_m})\nonumber
\end{gather}
is a convex function of $J(h)$ for fixed $P(h)$,
and a concave function of $P(h)$ for fixed $J(h)$.
These properties imply that there exists at least one saddle point of the game.

Writing the KKT conditions for both optimization problems we get \cite{uluk1}:
\begin{gather}\label{app1.1}
-\frac{h}{\sigma_N^2+J(h)+hP(h)}+\lambda-\gamma(h)=0
\end{gather}
and
\begin{gather}\label{app1.2}
-\frac{hP(h)}{(\sigma_N^2+J(h))(\sigma_N^2+J(h)+hP(h))}+\nu-\delta(h)=0,
\end{gather}
where $\gamma(h)$ and $\delta(h)$ are the complementary slackness variables for $P(h)$ and $J(h)$, respectively.

The three possible cases are \cite{uluk1}: Case 1: $P(h)>0$, $J(h)>0$; Case 2:
$P(h)>0$, $J(h)=0$ and Case 3: $P(h)=J(h)=0$.

For Case 1 both complementary slackness variables are $0$, and solving (\ref{app1.1}) and
(\ref{app1.2}) together we get 
\begin{gather}\label{app1.3}
\frac{\lambda}{\mu}=\frac{J(h)+\sigma_N^2}{P(h)},
\end{gather}
and
\begin{gather}\label{app1.4}
P(h)=\frac{h}{\lambda(h+\lambda/\mu)},
\end{gather}
while for Cases 2 and 3, the solution is plain water-filling \cite{uluk1}. 

These considerations result in the solutions (\ref{sol11})
and (\ref{sol12}).


\subsection{Proof of Theorem \ref{thm_short_term}}
This proof follows the one described in the Appendix B of \cite{caire}.
The probability of outage can be written as:
\begin{gather}\label{ap11}
Pr(I_M(\mathbf{h},P(h),J(h))<R)=E[\chi_{\{  I_M(\mathbf{h},P(h),J(h))<R \}}],
\end{gather}
where $\chi_{\{\mathscr{A}\}}$ denotes the indicator function of the set $\mathscr{A}$.
Replacing the power allocations by the solutions of the game described by (\ref{game21}) and (\ref{game22}), we define
\begin{gather}
\chi^*(\mathbf{h})=\chi_{\{  I_M(\mathbf{h},P^*(h),J^*(h))<R \}}.
\end{gather}

Then the region $\mathcal{U}(R,\mathcal{P}, \mathcal{J})$ can be written as:
\begin{gather}
\mathcal{U}(R,\mathcal{P}, \mathcal{J})=\{ \mathbf{h}\in \mathbb{R}_{+}^M : \chi^*(\mathbf{h})=0\}.
\end{gather}

We next use the fact that the pair $(P^*(h), J^*(h))$ determines an equilibrium of the game (\ref{game21}), (\ref{game22}).
Thus, for any random power allocation $P(h)$ satisfying the power constraint, we can write:
\begin{gather}\label{ap12}
\chi^*(\mathbf{h})\leq \chi_{\{  I_M(\mathbf{h},P(h),J^*(h))<R \}} , \textrm{with probability 1}.
\end{gather}
Similarly, for any random $J(h)$, we have
\begin{gather}\label{ap13}
\chi^*(\mathbf{h})\geq \chi_{\{  I_M(\mathbf{h},P^*(h),J(h))<R \}} , \textrm{with probability 1}.
\end{gather}

Now pick some arbitrary power allocation functions $P_a(h)$ and $J_a(h)$, which satisfy the short-term power constraints,
and set 
\begin{gather}
\widehat{P}(h)=(1-\chi^*(\mathbf{h}))P^*(h)+\chi^*(\mathbf{h})P_a(h),
\end{gather}
and
\begin{gather}
\widehat{J}(h)=(1-\chi^*(\mathbf{h}))J_a(h)+\chi^*(\mathbf{h})J^*(h),
\end{gather}
It is easy to see that $1/M\sum_{m=0}^{M-1}\widehat{P}(h_m)\leq \mathcal{P}$ with probability $1$ ,
$1/M\sum_{m=0}^{M-1}\widehat{J}(h_m)\leq \mathcal{J}$ with probability $1$, and moreover that
\begin{gather}
\chi^*(\mathbf{h})=\chi_{\{  I_M(\mathbf{h},\widehat{P}(h),\widehat{J}(h))<R \}}.
\end{gather}

Note that transmitter and jammer could pick $P_a(h)=0$ and $J_a(h)=0$ respectively,
but this strategy would not improve their performances (power cannot be saved),
since the only power constraints are set over frames.

Now, using (\ref{ap11}), (\ref{ap12}) and (\ref{ap13}), we get:
\begin{gather}
Pr(I_M(\mathbf{h},P(h),\widehat{J}(h))<R)\geq {}\nonumber\\
{}\geq Pr(I_M(\mathbf{h},\widehat{P}(h),\widehat{J}(h))<R)\geq{} \nonumber\\
{}\geq Pr(I_M(\mathbf{h},\widehat{P}(h),J(h))<R),
\end{gather}
which proves the existence of a Nash equilibrium of the original game.


\section{Long-Term Power Constraints: Pure Strategies}\label{app2}

\subsection{Proof of Proposition \ref{circ_pr_prop1}}

Take \emph{Problem 1}.
Let $(\mathfrak{P}^*,\mathfrak{J}^*)=\big((P_0^*,P_1^*,\ldots, P_{M-1}^*),(J_0^*,J_1^*,\ldots,
J_{M-1}^*)\big)$ be a solution such that $\sum_{m=0}^{M-1}P_m^*=P_{M,1}$ and
$\sum_{m=0}^{M-1}J_m^*=J_{M,1}$,
and assume that $I_M(\mathfrak{P}^*,\mathfrak{J}^*)>R$.
Since $I_M$ is a continuous, strictly increasing function of $P_0$, without loss of generality, we can find $P_0'<P_0^*$ such that
$I_M((P_0',P_1^*,\ldots, P_{M-1}^*),\mathfrak{J}^*)=R$.

But then $P_0'+\sum_{m=1}^{M-1} P_m^*<MP_{M,1}$, which means that $(\mathfrak{P}^*,\mathfrak{J}^*)$ is
suboptimal (from the transmitter's point of view), and hence not a solution.

Therefore, the first constraint $I_M\geq R$ has to be satisfied with equality, i.e. $I_M=R$.

Now take the solution $(\mathfrak{P}^*,\mathfrak{J}^*)$, and assume that $\frac{1}{M}\sum_{m=0}^{M-1} J_m^*<J_M$.
Then we can find $J_0'>J_0^*$, such that $J_0'+\sum_{m=1}^{M-1} J_m^*=MJ_M$.
In order for the first constraint $I_M=R$ to be satisfied, the value and distribution of $P_M$ will have to be modified.

We prove next that the value of $P_M$ should be increased, which makes the pair $(\mathfrak{P}^*,\mathfrak{J}^*)$ suboptimal
(from the jammer's point of view), thus
contradicting the hypothesis that it is a solution, and proving that the second constraint should hold with equality.

Assume there is a distribution $\mathfrak{P}''=(P_0'',P_1'',\ldots, P_{M-1}'')$ that
minimizes $P_M$, under the constraint $I_M(\{P_m\},(J_0',J_1^*, \ldots, J_{M-1}^*))=R$, such that 

\begin{gather}\label{contra}
\sum_{m=0}^{M-1} P_m'' \leq P_{M,1}.
\end{gather}
 
Then, replacing $J_0$ by its old value $J_0^*$, we have that $(\mathfrak{P}'',\mathfrak{J}^*)$ is either a second solution of Problem 1
(if (\ref{contra}) is satisfied with equality), or a better choice (if (\ref{contra}) is satisfied with strict inequality).
We can readily dismiss the latter case.
For the former case, $I_M$ is a strictly decreasing function of $J_0$, thus $I_M(\mathfrak{P}'',\mathfrak{J}^*)>R$,
which contradicts the first part of this proof.
The same arguments work for \emph{Problem 2}.


\subsection{Proof of Proposition \ref{circ_pr_thm}}

Proposition \ref{circ_pr_thm} is a direct consequence of Theorem 8 in the Appendix II.D of \cite{myself3}.
We restate the theorem here for completeness. For a complete proof, see the first part of this paper \cite{myself3}.

\begin{thm}
Take $x,y \in L^2[\mathbb{R}]$ and define the order relation $x>y$ if and only if $x(t)>y(t)~\forall
t\in \mathbb{R}$. 
Consider the continuous real functions $f(x)$, $g(y)$ and $h(x,y)$ over $L^2[\mathbb{R}]$, such that
$f$ is a strictly increasing function of $x$, $g$ is a strictly increasing function of $y$,
and $h$ is a strictly increasing function of $x$ for fixed $y$ and
a strictly decreasing function of $y$ for fixed $x$.

Define the following minimax and maximin problems:
\be \label{genpr1p2}
\max_{y\geq0}\left[ \min_{x\geq0} f(x) ~\textrm{s.t.}~ h(x,y)\geq H \right]
\textrm{s.t.} g(y)\leq G,
\ee

\be \label{genpr2p2}
\max_{x\geq0}\left[ \min_{y\geq0} g(y) ~\textrm{s.t.}~ h(x,y)\leq H \right]
\textrm{s.t.} f(x)\leq F,
\ee

\be \label{genpr3p2}
\min_{y\geq0}\left[ \max_{x\geq0} h(x,y) ~\textrm{s.t.}~ f(x)\leq F \right]
\textrm{s.t.} g(y)\leq G.
\ee

(I) Choose any real values for $G$ and $H$.
Take problem (\ref{genpr1p2}) under these constraints and let the pair $(x_1,y_1)$ denote one of its optimal solutions,
yielding a value of the objective function $f(x_1)=F_1$. If we set the value of the corresponding
constraints in problems (\ref{genpr2p2}) and (\ref{genpr3p2}) to $F=F_1$, then the values of the objective
functions of problems (\ref{genpr2p2}) and (\ref{genpr3p2}) under their optimal solutions are
$g(y)=G$ and $h(x,y)=H$, respectively. Moreover, $(x_1,y_1)$ is also an optimal solution of all problems.

(II) Choose any real values for $F$ and $H$.
Take problem (\ref{genpr2p2}) under these constraints and let the pair $(x_2,y_2)$ denote one of its optimal solutions,
yielding a value of the objective function $g(y_2)=G_2$. If we set the value of the corresponding
constraints in problems (\ref{genpr1p2}) and (\ref{genpr3p2}) to $G=G_2$, then the values of the objective
functions of problems (\ref{genpr1p2}) and (\ref{genpr3p2}) under their optimal solutions are
$f(x)=F$ and $h(x,y)=H$, respectively. Moreover, $(x_2,y_2)$ is an optimal solution of all problems.
 
(III) Choose any real values for $F$ and $G$.
Take problem (\ref{genpr3p2}) under these constraints and let the pair $(x_3,y_3)$ denote one of its optimal solutions,
yielding a value of the objective function $h(x_3,y_3)=H_3$. If we set the value of the corresponding
constraints in problems (\ref{genpr1p2}) and (\ref{genpr2p2}) to $H=H_3$, then the values of the objective
functions of problems (\ref{genpr1p2}) and (\ref{genpr2p2}) under their optimal solutions are
$f(x)=F$ and $g(y)=G$, respectively. Moreover, $(x_3,y_3)$ is an optimal solution of all problems.
\end{thm}


\subsection{Proof of Proposition \ref{uniqueness_prop}}\label{app23}

Take \emph{Problem 1}.
By Proposition \ref{circ_pr_thm}, if there exists $P_{M,1}$ such that solving the game in (\ref{game21}) and (\ref{game22})
with the constraint $\sum_{m=1}^{M-1}P_m \leq M P_{M,1}$ yields the objective $I_M(\mathbf{h},\{P_m\},\{J_m\})=R$, then the solution
of \emph{Problem 1} coincides with the solution of the game in (\ref{game21}) and (\ref{game22}).

We write this solution as in (\ref{sol11}) and (\ref{sol12}), but we denote $\lambda=1/\eta$ and
$\mu= \nu/\eta$:
\begin{gather}\label{solution1}
P_m^*=\left\{ \begin{array}{ccc} (\lambda-\frac{\sigma_N^2}{h_m})^+ & \textrm{if} & h_m<\frac{\sigma_N^2}{\lambda-\sigma_N^2 \mu}\\
\mu\frac{\lambda h_m}{1+\mu h_m} & \textrm{if} & h_m\geq\frac{\sigma_N^2}{\lambda-\sigma_N^2 \mu}
\end{array} \right.
\end{gather}

\begin{gather}\label{solution2}
J_m^*=\left\{ \begin{array}{ccc} 0 & \textrm{if} & h_m<\frac{\sigma_N^2}{\lambda-\sigma_N^2 \mu}\\
\frac{\lambda h_m}{1+\mu h_m}-\sigma_N^2 & \textrm{if} & h_m\geq\frac{\sigma_N^2}{\lambda-\sigma_N^2 \mu}
\end{array} \right.
\end{gather}
where $\lambda$ and $\mu$ are constants that can be determined from the constraints
$\sum_{m=1}^{M-1}J_m=M J_M$ and $\sum_{m=1}^{M-1}I(h_m,P_m,J_m)=MR$.

We shall use the following conventions and denotations:
\begin{itemize}
\item Without loss of generality, we shall assume that the blocks in a frame are indexed in increasing order of
their channel coefficients. That is, $h_0\leq h_1 \leq\ldots, \leq h_{M-1}$.
\item Denote $x_m=J_m+\sigma_N^2$ and $x_m^*=J_m^*+\sigma_N^2$. Note that $\frac{x_0^*}{h_0}\geq \frac{x_1^*}{h_1}
\geq\ldots, \geq \frac{x_{M-1}^*}{h_{M-1}}$.
\item Denote by $h_p$ the first block on which the transmitter's power is strictly positive, and by $h_j$ the
first block on which the jammer's power is strictly positive. Note that $h_p\leq h_j$.
\end{itemize}
Note that
\be\label{Pmfirstexpr}
P_m^*=\left[\lambda-\frac{x_m^*}{h_m}\right]_+
\ee
for all $m\in\{0,1,\ldots,M-1\}$, where $[z]_+=\max\{z,0\}$.

Given these and (\ref{solution1}) and (\ref{solution2}) above, we can write:

\begin{eqnarray}\label{sys10d}
\frac{\sigma_N^2}{h_p}\leq\lambda < \frac{\sigma_N^2}{h_{p-1}},
\end{eqnarray}
\begin{eqnarray}\label{sys20d}
\sigma_N^2\frac{1+\mu h_{j}}{h_j}\leq\lambda < \sigma_N^2\frac{1+\mu (h_{j-1})}{h_{j-1}},
\end{eqnarray}
\begin{eqnarray}\label{sys30d}
MR=\sum_{m=p}^{j-1}\log\left(\frac{\lambda h_m}{\sigma_N^2}\right)+{}\nonumber\\
{}-\sum_{m=j}^{M-1}\log\left(\frac{1}{1+\mu h_m}\right),
\end{eqnarray}

Denote by $Q_U[h]$ denotes the index of the smallest channel coefficient in the frame that is larger than $h$.
With this notation, we can write

\be
p\geq Q_U\left[\frac{h_{j-1}}{1+\mu h_{j-1}}\right]
\ee
\be
h_{p-1}<\frac{h_{j}}{1+\mu h_{j}}
\ee

\begin{eqnarray}\label{sys2d}
\frac{1}{M}\sum_{m=j}^{M-1}\left[\frac{\frac{h_m}{1+\mu h_m}}{\frac{h_j}{1+\mu h_j}}-1\right]\leq\frac{J_M}{\sigma_N^2} \leq \nonumber\\
\frac{1}{M}\sum_{m=j}^{M-1}\left[\frac{\frac{h_m}{1+\mu h_m}}{\frac{h_{j-1}}{1+\mu h_{j-1}}}-1\right],
\end{eqnarray}

\begin{eqnarray}\label{sys3d}
\sum_{m=Q_U\left[\frac{h_{j}}{1+\mu h_{j}}\right]}^{j-1}\log\left(h_m\frac{1+\mu h_j}{h_j}\right)-{}\nonumber\\
{}-\sum_{m=j}^{M-1}\log\left(\frac{1}{1+\mu h_m}\right)\leq MR \leq{}\nonumber\\
{}\leq\sum_{m=Q_U\left[\frac{h_{j-1}}{1+\mu h_{j-1}}\right]}^{j-1}\log\left(h_m\frac{1+\mu(h_{j-1})}{h_{j-1}}\right)-{}\nonumber\\
{}-\sum_{j}^{M-1}\log\left(\frac{1}{1+\mu h_m}\right),
\end{eqnarray}
where (\ref{sys2d}) follows from $J_M=\sum_{m=j}^{M-1}\left[\frac{\lambda h_m}{1+\mu h_m}-\sigma_N^2\right]$, 
and the first inequality in (\ref{sys3d}) follows since $h_{p-1}< \frac{h_{j}}{1+\mu h_{j}}$
implies $p\leq Q_U\left[\frac{h_{j}}{1+\mu h_{j}}\right]$ because there is no other channel coefficient
between $h_{p-1}$ and $h_p$.

It is straightforward to show that for fixed $h_j$ the left-most and the right-most terms of inequality (\ref{sys2d})
are strictly decreasing functions of $\mu$, while the left-most and the right-most terms of inequality (\ref{sys3d})
are strictly increasing functions of $\mu$.

Note that
\be
\sum_{m=j}^{M-1}\left[\frac{\frac{h_m}{1+\mu h_m}}{\frac{h_j}{1+\mu h_j}}-1\right]=
\sum_{m=j+1}^{M-1}\left[\frac{\frac{h_m}{1+\mu h_m}}{\frac{h_j}{1+\mu h_j}}-1\right],
\ee
and
\be
\sum_{m=Q_U\left[\frac{h_{j}}{1+\mu h_{j}}\right]}^{j-1}\log\left(h_m\frac{1+\mu h_j}{h_j}\right)-{}\nonumber\\
{}-\sum_{m=j}^{M-1}\log\left(\frac{1}{1+\mu h_m}\right)={}\nonumber\\
{}=\sum_{m=Q_U\left[\frac{h_{j}}{1+\mu h_{j}}\right]}^{j}\log\left(h_m\frac{1+\mu h_j}{h_j}\right)-{}\nonumber\\
{}-\sum_{m=j+1}^{M-1}\log\left(\frac{1}{1+\mu h_m}\right).
\ee
That is, by keeping $\mu$ constant and replacing $h_j$ by $h_{j-1}$ in both first terms of
(\ref{sys2d}) and (\ref{sys3d}), we get exactly the last terms of (\ref{sys2d}) and (\ref{sys3d}), respectively.

Finally, we take a contradictory approach. Suppose there exist two different pairs $(h_{j1},\mu_1)$ and
$(h_{j2},\mu_2)$ that satisfy both (\ref{sys2d}) and (\ref{sys3d}) and assume, without loss of generality
that $h_{j1}<h_{j2}$. Then, in order for $(h_{j2},\mu_2)$ to satisfy (\ref{sys2d}) we need $\mu_2>\mu_1$, while
in order for $(h_{j2},\mu_2)$ to satisfy (\ref{sys3d}) we need $\mu_<\mu_1$.
Thus $h_j$ is unique. Note however that the relations above do not guarantee the uniqueness of $\mu$.

For the optimal $h_j$, the constraint $\sum_{m=1}^{M-1}J_m=M J_M$ translates to
\be\label{newapproach1}
\sum_{m=j}^{M-1} \frac{\lambda h_m}{1+\mu h_m}=MJ_M +(M-j)\sigma_N^2.
\ee
while the constraint $I_M(\mathbf{h},\{P_m\},\{J_m\})=R$ is already given in (\ref{sys30d}).
The left hand side of (\ref{newapproach1}) is a strictly increasing function of $\lambda$ for fixed $\mu$
and a strictly decreasing function of $\mu$ for fixed $\lambda$, while being equal to a constant.

Again, for a contradictory approach, suppose there exist two different pairs of $(\mu_1,\lambda_1)$ and
$(\mu_2,\lambda_2)$ that can generate different solutions. If we assume, without loss of generality that
$\mu_1>\mu_2$, then, in order for (\ref{newapproach1}) to be satisfied by both pairs, we need $\lambda_1>\lambda_2$.
But this can only mean that under $(\mu_2,\lambda_2)$ the transmitter allocates non-zero power
to more channel coefficients than under $(\mu_1,\lambda_1)$. This remark simply says that the index $p$ at which
the transmitter starts transmitting is a decreasing function of $\lambda$, and can easily be verified by
(\ref{Pmfirstexpr}).  

Looking now at (\ref{sys30d}), we observe that its right hand side is a strictly increasing function of $\lambda$
for fixed $\mu$ and a strictly increasing function of $\mu$ for fixed $\lambda$, while being equal to a constant.
In other words, if (\ref{sys30d}) is satisfied by the pair $(\mu_1,\lambda_1)$, then it cannot also be satisfied
by $(\mu_2,\lambda_2)$.
Thus, the pair $(\lambda, \mu)$ that satisfies both (\ref{sys30d}) and (\ref{newapproach1}) is also unique.
But once $h_j$, $\lambda$ and $\mu$ are given, $h^p$ is uniquely determined.
Therefore there cannot exist more than one solution to \emph{Problem1}.

Similar arguments can be applied to show that the solution of \emph{Problem2} is unique.


\subsection{Proof of Proposition \ref{propconcave}}

Since the solution is unique, it follows that $\mathscr{P}_M(J_M)$ is a strictly increasing function.
By closely inspecting the form of the solution in (\ref{solution1}) and (\ref{solution2}),
it is straightforward to see that if $J_M\to \infty$, then $J_m\to\infty$ for all $m\in\{0,1,\ldots, M-1\}$.
If the required $P_M$ were finite, this would imply $I_M\to 0$, which violates the power constraints of
\emph{Problem 1}.

For \emph{Problem 1} we prove that the resulting $\mathscr{P}_M(J_M)$ function is continuous and concave in several steps.
We first show in Lemma \ref{propapp51} that
the optimal jammer strategy $\{x_m^*\}_{m=0}^{M-1}$ is a continuous function of
the given jamming power $J_M$.
Lemma \ref{propapp52} proves that $P_M(\{x_m\})$ is continuous and has
continuous first order derivatives.
This implies that $P_M(J_M)$ is in fact continuous and has a continuous first order derivative.
Finally, Lemma \ref{propapp53} shows that for any fixed $h_p$ and $h_j$ the function $P_M(J_M)$ is concave.

\vspace*{4pt}
\begin{lemma}\label{propapp51}
The optimal jammer power allocation $\{x_m^*\}_{m=0}^{M-1}$ within a frame
is a continuous increasing function of the given jamming power $J_M$ invested over that frame.
\end{lemma}
\vspace*{4pt}
\begin{proof}
It is clear that $x_m^*$ is continuous and increasing as a function of $J_M$ if $h_p$ and $h_j$ are fixed.
At any point where either $h_p$ or $h_j$ change as a result of a change in $J_M$, the optimal jamming strategy
$\{x_m^*\}_{m=0}^{M-1}$ maintains continuity as a result of the uniqueness of the solution (Proposition \ref{uniqueness_prop}).
\end{proof}
\vspace*{4pt}

\begin{lemma}\label{propapp52}
Both $P_M(\{x_m\})$ and the derivatives $\frac{dP_M}{dx_m}$ are continuous
functions of $\{x_m\}_{m=0}^{M-1}$.
\end{lemma}
\vspace*{4pt}
\begin{proof}

Consider any two points $\mathfrak{X}_1=(x_{1,m})_{m=0}^{M-1}$ and $\mathfrak{X}_2=(x_{2,m})_{m=0}^{M-1}$ and any
trajectory $\mathfrak{T}$ that connects them.

For a given vector $\mathfrak{X}=(x_{m})_{m=0}^{M-1}$, the required transmitter power is
\begin{eqnarray}\label{relPM}
P_M=\frac{M-p}{M}\left(\frac{c}{\left(\prod_{m=p}^{M-1}h_m\right)}\right) ^{\frac{1}{M}}
\left(\prod_{m=p}^{M-1}x_m\right)^{\frac{1}{M}}-{}\nonumber\\
{}-\frac{1}{M}\sum_{m=p}^{M-1}\frac{x_m}{h_m}.
\end{eqnarray}
Note that $p$ depends upon the choice of $\mathfrak{X}$.
For fixed $p$, the continuity and differentiability of $P_M(\mathfrak{X})$ are obvious.
Thus, it suffices to show that these properties also hold in a point of $\mathfrak{T}$ where $p$ changes.

If we can show continuity and differentiability when $p$ is decreased by $1$, then larger variations of $p$
can be treated as multiple changes by $1$, and continuity still holds.

Recall the assumption that the channel coefficients are always indexed
in decreasing order of the quantities $\frac{x_m}{h_m}$.
Let $\mathfrak{X}_k=(x_{k,m})_{m=0}^{M-1}$ be a point of $\mathfrak{T}$ where
the transmitter decreases the index of the block over which it starts to transmit from $p_k$ to $p_k-1$,
and denote by $\mathfrak{T_1}$ the part of the trajectory $\mathfrak{T}$ that is between
$\mathfrak{X}_1$ and $\mathfrak{X}_k$, and $\mathfrak{T_2}=\mathfrak{T}\setminus \mathfrak{T_1}$.

Since $P_{p_k-1}=0$, we know that $\lambda$ does not change in this point, since
\be
\frac{1}{M}\sum_{m=p}^{M-1}\left[\lambda-\frac{x_m}{h_m}\right]=
\frac{1}{M}\sum_{m=p-1}^{M-1}\left[\lambda-\frac{x_m}{h_m}\right]=P_M.
\ee

Define the ``left'' and ``right'' limits  $P_M(\mathfrak{X}_k-)$ and $P_M(\mathfrak{X}_k+)$ as:
\begin{gather}\label{leftrightlimit1}
P_M(\mathfrak{X}_k-)=\lim_{\substack{\mathfrak{X}\to \mathfrak{X}^k\\ \mathfrak{X}\in
\mathfrak{T_1}}} P_M(\mathfrak{X}),
\end{gather}
\begin{gather}\label{leftrightlimit2}
P_M(\mathfrak{X}_k+)=\lim_{\substack{\mathfrak{X}\to \mathfrak{X}^k\\ \mathfrak{X}\in
\mathfrak{T_2}}} P_M(\mathfrak{X}).
\end{gather}
Since $\mathbb{R}_+^M$ is Hausdorff \cite{munkres}, there exists a small enough neighborhood $\mathfrak{U}\subset \mathbb{R}_+^M$
of $\mathfrak{X}_k$, such that
$p(\mathfrak{X})=p_k$ to the ``left'' and $p(\mathfrak{X})=p_k-1$ to the ``right'' of $\mathfrak{X}_k$ on $\mathfrak{U}$.
We can now write: 
\begin{gather}
P_M(\mathfrak{X}_k+)={}\nonumber\\
{}=\lambda \frac{M-p_k+1}{M}-
\frac{1}{M}\sum_{m=p_k-1}^{M-1}\frac{x_{k,m}}{h_{m}}={}\nonumber\\
{}=\lambda \frac{M-p_k}{M}-
\frac{1}{M}\sum_{m=p_k}^{M-1}\frac{x_{k,m}}{h_{m}}+{}\nonumber\\
{}+\frac{1}{M}(\lambda-\frac{x_{k,p_k-1}}{h_{p_k-1}})=P_M(\mathfrak{X}_k-),
\end{gather}
where the last equality follows because $\lambda=\frac{x_{k,p_k-1}}{h_{p_k-1}}$.
This proves continuity.

Similar arguments can be used to show the continuity of the derivatives
\begin{gather}
\frac{dP_M}{dx_{n}}=\frac{1}{M}\left(\frac{\lambda}{x_{n}}-\frac{1}{h_n}\right)
\end{gather}
in $\mathfrak{X}_k$ (note that $\frac{\lambda}{x_{k,p_k-1}}=\frac{1}{h_{p_k-1}}$).

Therefore, $P_M(\mathfrak{X})$ is continuous and has first-order derivatives that are continuous along any
trajectory $\mathfrak{T}$ between any two points $\mathfrak{X}_1$ and $\mathfrak{X}_2$.
\end{proof}
\vspace*{4pt}

Finally, for the last part of our proof:
\vspace*{4pt}
\begin{lemma}\label{propapp53}
For fixed $p$ and $j$, the function $P_M(J_M)$ is concave.
\end{lemma}
\vspace*{4pt}
\begin{proof}

We can write
\begin{eqnarray}\label{reljm}
MJ_M+(M-j)\sigma_N^2=\Bigg[c\prod_{m=p}^{j-1}\left(\frac{\sigma_N^2}{h_m}\right)^{\frac{1}{M}}\cdot\nonumber\\
\cdot\prod_{m=j}^{M-1}\left(\frac{1}{1+\mu h_m}\right)^{\frac{1}{M}}\Bigg]^{\frac{M}{j-p}}
\sum_{m=j}^{M-1}\frac{h_m}{1+\mu h_m},
\end{eqnarray}
and denote
\begin{gather}
g(\mu)=\prod_{m=j}^{M-1}\left(\frac{1}{1+\mu h_m}\right)^{\frac{1}{j-p}}\sum_{m=j}^{M-1}\frac{h_m}{1+\mu h_m}
\end{gather}
Note that for fixed $p$ and $j$, $J_M$ is a linear function of $g$.

A similar relation can be found for the required transmitter power $P_M$:
\begin{eqnarray}\label{relpm}
MP_M+\frac{1}{M}\sum_{m=p}^{j-1}\frac{\sigma_N^2}{h_m}=
\Bigg[c\prod_{m=p}^{j-1}\left(\frac{\sigma_N^2}{h_m}\right)^{\frac{1}{M}}\cdot\nonumber\\
\cdot\prod_{m=j}^{M-1}\left(\frac{1}{1+\mu h_m}\right)^{\frac{1}{M}}\Bigg]^{\frac{M}{j-p}}\cdot{}\nonumber\\
{}\cdot \left[\frac{M-p}{M}-\frac{1}{M}\sum_{m=j}^{M-1}\frac{1}{1+\mu h_m}\right].
\end{eqnarray}
Denote
\begin{gather}
f(\mu)=\prod_{m=j}^{M-1}\left(\frac{1}{1+\mu h_m}\right)^{\frac{1}{j-p}}\cdot{}\nonumber\\
{}\cdot \left[(M-p)-\sum_{m=j}^{M-1}\frac{1}{1+\mu h_m}\right]
\end{gather}
and note that for fixed $p$ and $j$, $P_M$ is a linear function of $f$.

It suffices to show that $f(g)$ is concave.
For this purpose, the derivative $\frac{df}{dg}=\frac{df}{d\mu}(\frac{d\mu}{dg})^{-1}$ should be
a decreasing function of $g$, and hence an increasing function of $\mu$.

Computing the derivatives from (\ref{reljm}) and (\ref{relpm}) we get:
\begin{eqnarray}\label{relg}
\frac{df}{dg}=\frac{\frac{df}{d\mu}}{\frac{dg}{d\mu}}={}\nonumber\\
{}=\frac{\frac{1}{j-p}\left((M-p)-\sum_{m=j}^{M-1}\frac{1}{1+\mu h_m}\right)
-\frac{\sum_{m=j}^{M-1}\frac{h_m}{(1+\mu h_m)^2}}{\sum_{m=j}^{M-1}\frac{h_m}{1+\mu h_m}}}
{\frac{1}{j-p}\sum_{m=j}^{M-1}\frac{h_m}{(1+\mu h_m)^2}+
\frac{\sum_{m=j}^{M-1}\frac{h_m^2}{(1+\mu h_m)^2}}{\sum_{m=j}^{M-1}\frac{h_m}{1+\mu h_m}}}
\end{eqnarray}

Arguments similar to those in \cite{myself3} apply in proving that above the derivative increases with $\mu$.
Looking at the right hand side of (\ref{relg}) (the ``large fraction''), we notice that the
first term in the numerator increases with $\mu$.
For the second term in the numerator, it is clear that as $\mu$ increases, its numerator decreases
faster than its denominator. This implies that the whole numerator of the ``large fraction'' is an increasing function of $\mu$.
Similarly, the first term in the denominator is clearly a decreasing function of $\mu$.
The only thing left is the second term of the denominator.
It is straightforward to show that its derivative with respect to $\mu$
can be written as
\be\label{jvouscblajkghilu1}
\frac{d}{d\mu}\frac{\sum_{m=j}^{M-1}\frac{h_m^2}{(1+\mu h_m)^2}}{\sum_{m=j}^{M-1}\frac{h_m}{1+\mu h_m}}
=\frac{1}{\left[\sum_{m=j}^{M-1}\frac{h_m}{1+\mu h_m}\right]^2}\cdot\nonumber\\
\cdot\Bigg\{ \left[\sum_{m=j}^{M-1}\frac{h_m^2}{(1+\mu h_m)^2}\right]^2-\sum_{m=j}^{M-1}\frac{h_m^3}{(1+\mu h_m)^3}\cdot\nonumber\\
\cdot \sum_{m=j}^{M-1}\frac{h_m}{(1+\mu h_m)}\Bigg\}
\ee

If we consider the fact that for any two real numbers $a$ and $b$ we have
\be
(a^2+b^2)^2-(a+b)(a^3+b^3)=-ab(a-b)^2
\ee 
and the summations in (\ref{jvouscblajkghilu1}) are positive, it is easy to see that the second term of the denominator
of the ``large fraction'' is decreasing with $\mu$. 
Hence overall the derivative in (\ref{relg}) increases with $\mu$.

\end{proof}
\vspace*{4pt}


\section{Long Term Power Constraints: Mixed Strategies}\label{app3}

\subsection{Proof of Theorem \ref{thm1_long_term_mixed_sf}}

Denote the solution of the game in (\ref{game21}) and (\ref{game22}), where the jammer is constrained
to $\frac{1}{M}\sum_{m=1}^{M-1}J_m \leq J_M(p_M)$ and the transmitter is constrained to
$\frac{1}{M}\sum_{m=1}^{M-1}P_m \leq p_M$ by $(\{P_{m,1}\},\{J_{m,1}\})$, and the solution of the game in
(\ref{game21}) and (\ref{game22}), where the
transmitter is constrained to $\frac{1}{M}\sum_{m=1}^{M-1}P_m \leq P_M(j_M)$ and the jammer is constrained to
$\frac{1}{M}\sum_{m=1}^{M-1}J_m \leq j_M$ by $(\{P_{m,2}\},\{J_{m,2}\})$.

Denote the solution of the game in (\ref{game21}) and (\ref{game22}), where the jammer is constrained
to $\frac{1}{M}\sum_{m=1}^{M-1}J_m\leq j_M$ and the transmitter is constrained to $\frac{1}{M}\sum_{m=1}^{M-1}P_m\leq p_M$
by $(\{P_{m,0}\},\{J_{m,0}\})$..

By the Proposition \ref{circ_pr_prop1}, we must have
$I_M(\{P_{m,1}\},\{J_{m,1}\})=R$ and $I_M(\{P_{m,2}\},\{J_{m,2}\})=R$,
where $I_M(\{P_m\},\{J_m\})=\frac{1}{M}\sum_{m=0}^{M-1}\log(1+\frac{P_m h_m}{J_m+\sigma_N^2})$.

We will show that (i) even if the jammer's power $j_M$ is different from $J_M(p_M)$,
the transmitter's strategy is still optimal; (ii) even if the transmitter's power
$p_M$ is different from $P_M(j_M)$, the jammer's strategy is still optimal.

Assume the transmitter plays the strategy given by $\{P_{m,1}\}$. 

If $j_M=J_M(p_M)$, it is clear that the optimal solution for both transmitter and jammer is the solution
of the game in (\ref{game21}) and (\ref{game22}), where the jammer is constrained
to $\frac{1}{M}\sum_{m=1}^{M-1}J_m\leq j_M$ and the transmitter is constrained to $\frac{1}{M}\sum_{m=1}^{M-1}P_m\leq p_M$. In this case,
it is as if each player knows the other player's power constraint.

If $j_M<J_M(p_M)$, then by Lemma \ref{propapp51} we have that $J_{m,0}<J_{m,1}~ \forall m$. Since $I_M(\{P_m\},\{J_m\})$
is a strictly decreasing function of $\{J_m\}$ (under the order relation defined in the Appendix III of \cite{myself3}),
this implies that $I_M(\{P_{m,1}\},\{J_{m,1}\})>R$. Note that $\{J_{m,0}\}$ is the jammer's strategy when the jammer knows
the transmitter's power constraint $p_M$.
Thus we have shown that when the transmitter plays $\{P_{m,1}\}$ and $j_M<J_M(p_M)$, the jammer cannot induce outage
over the frame even if it knew the value of $p_M$.

Assume that the jammer plays the strategy given by $\{J_{m,2}\}$. A similar argument shows that if $p_M<P_M(j_M)$, or
equivalently $j_M>J_M(p_M)$, the transmitter cannot achieve reliable communication over
the frame even if it knew the exact value of $j_M$.

This shows that $(\{P_{m,1}\}, \{J_{m,2}\})$ is a Bayes equilibrium \cite{meyerson} for the game with incomplete information
describing the power allocation within a frame.


\subsection{Proof of Proposition \ref{prop_saddlepoints_kkt}}

Take any solution $\{P_M(h)^*\},\{J_M(h)^*\}$ of the KKT conditions and denote by $P_{out}^*$ the
outage probability obtained under these strategies.
By maintaining $\{J_M(h)^*\}$ constant and changing $\{P_M(h)^*\}$, the resulting
probability of outage can only be greater than or equal to $P_{out}^*$, since the original $\{P_M(h)^*\}$
is the solution of a minimization problem with convex cost function and linear constraints.

Similarly, by maintaining $\{P_M(h)^*\}$ constant and changing $\{J_M(h)^*\}$, the resulting
probability of outage can only be less than or equal to $P_{out}^*$, since the original $\{J_M(h)^*\}$
is the solution of a maximization problem with concave cost function and linear constraints.

These arguments imply that $\{P_M(h)^*\},\{J_M(h)^*\}$ is a Nash equilibrium of the game.

\bibliographystyle{IEEEtran}
\bibliography{jamming}
\end{document}